\documentclass[draftclsnofoot, 12pt,onecolumn]{IEEEtran}
\IEEEoverridecommandlockouts
\usepackage{caption}
\usepackage{subcaption}
\usepackage{cite}
\usepackage{url}
\usepackage{amsmath,amssymb,amsfonts}
\usepackage{amsthm}
{
	\newtheorem{prob}{Problem}
	\newtheorem{corollary}{Corollary}
	
	\newtheorem{definition}{Definition}
	\newtheorem{lemma}{Lemma}
	\newtheorem{theorem}{Theorem}
	
}
\usepackage[T1]{fontenc}
\usepackage[english]{babel}
\usepackage{algorithm,algorithmicx,algpseudocode}
\usepackage{graphicx}
\usepackage{textcomp}
\usepackage{xcolor}
\def\BibTeX{{\rm B\kern-.05em{\sc i\kern-.025em b}\kern-.08em
    T\kern-.1667em\lower.7ex\hbox{E}\kern-.125emX}}
\begin{document}
\pagestyle{empty}
%\title{Paper Title*\\
%\and
%\and
%\and
%\and
%\and
\title{Optimizing Information Freshness via Multiuser Scheduling with Adaptive NOMA/OMA%: An MDP Approach %Timely Status Update in Resource-Constrained IoT Systems with Random Arrival: An MDP Approach
	%\thanks{Identify applicable funding agency he\right) re. If none, delete this.}
}
\author{Qian Wang, He Chen, Changhong Zhao, Yonghui Li, Petar Popovski and Branka Vucetic
	\thanks{ The work of H. Chen is supported by the CUHK direct grant under the project code 4055126. Part of the paper was presented on IEEE ISIT 2020 \cite{qian2020minimizing}.}
	\thanks{Q.Wang is with  School of Electrical
		and Information Engineering, The University of Sydney, Sydney, NSW
		2006, Australia and Department of Information Engineering, The Chinese University of Hong Kong, Hong Kong SAR, China. The work is done when she is a visiting student at CUHK (email:qian.wang2@sydney.edu.au).}%
	\thanks{H. Chen and C. Zhao are with Department of Information Engineering, The Chinese University of Hong Kong, Hong Kong SAR, China (email: \{he.chen, chzhao\}@ie.cuhk.edu.hk).}%
	\thanks{Y. Li and B. Vucetic are with School of Electrical
		and Information Engineering, The University of Sydney, Sydney, NSW
		2006, Australia (email: \{yonghui.li, branka.vucetic\}@sydney.edu.au).}
	\thanks{P. Popovski is with the Department of Electronic Systems, Faculty of Engineering and Science, APNet Section, Aalborg University, 9220 Aalborg, Denmark (email: petarp@es.aau.dk)}
	%\thanks{}
%\author{\IEEEauthorblockN{Qian Wang\textsuperscript{1,2}, He Chen\textsuperscript{2}, Yonghui Li\textsuperscript{1}, Branka Vucetic\textsuperscript{1}}%
%	\thanks{This work is done when Qian Wang is a visiting student at the Chinese University of Hong Kong.}
%	\IEEEauthorblockA{\textsuperscript{1}School of Electrical and Information Engineering, The University of Sydney, Sydney, Australia} 
%	%\{qian.wang2, he.chen, yonghui.li, branka.vucetic\}@sydney.edu.au}%
%	\IEEEauthorblockA{\textsuperscript{2} Department of Information Engineering, Chinese University of Hong Kong, Hong Kong SAR, China\\
%		%City, Country \\
%		\textsuperscript{1}\{qian.wang2,  yonghui.li, branka.vucetic\}@sydney.edu.au, \textsuperscript{2}he.chen@ie.cuhk.edu.hk}
	}

\maketitle
\thispagestyle{empty}
\begin{abstract}
This paper considers a wireless network with a base station (BS) conducting timely status updates to multiple clients via adaptive non-orthogonal multiple access (NOMA)/orthogonal multiple access (OMA). Specifically, the BS is able to adaptively switch between NOMA and OMA for the downlink transmission to optimize the information freshness of the network, characterized by the Age of Information (AoI) metric. If the BS chooses OMA, it can only serve one client within each time slot and should decide which client to serve; if the BS chooses NOMA, it can serve more than one client at the same time and needs to decide the power allocated to the served clients. For the simple two-client case, we formulate a Markov Decision Process (MDP) problem and develop the optimal policy for the BS to decide whether to use NOMA or OMA for each downlink transmission based on the instantaneous AoI of both clients. The optimal policy is shown to have a switching-type property with obvious decision switching boundaries. A near-optimal policy with lower computation complexity is also devised. For the more general multi-client scenario, inspired by the proposed near-optimal policy, we formulate a nonlinear optimization problem to determine the optimal power allocated to each client by maximizing the expected AoI drop of the network in each time slot. We resolve the formulated problem by approximating it as a convex optimization problem. We also derive the upper bound of the gap between the approximate convex problem and the original nonlinear, nonconvex problem. Simulation results validate the effectiveness of the adopted approximation. The performance of the adaptive NOMA/OMA scheme by solving the convex optimization is shown to be close to that of max-weight policy solved by exhaustive search. Besides, the adaptive NOMA/OMA scheme has achieved significant performance improvement comparing to the OMA scheme, especially when the number of clients in the network is large and the transmission SNR is high.
%This paper considers a wireless network with a base station (BS) conducting timely transmission to two clients in a slotted manner via hybrid non-orthogonal multiple access (NOMA)/orthogonal multiple access (OMA). Specifically, the BS is able to adaptively switch between NOMA and OMA for the downlink transmission to minimize the information freshness, characterized by Age of Information (AoI), of the network. If the BS chooses OMA, it can only serve one client within a time slot and should decide which client to serve; if the BS chooses NOMA, it can serve both clients simultaneously and should decide the power allocated to each client. To minimize the weighted sum of expected AoI of the network, we formulate a Markov Decision Process (MDP) problem and develop an optimal policy for the BS to decide whether to use NOMA or OMA for each downlink transmission based on the instantaneous AoI of both clients. We prove the existence of optimal stationary and deterministic policy, and perform action elimination to reduce the action space for lower computation complexity. The optimal policy is shown to have a switching-type property with obvious decision switching boundaries. A suboptimal policy with lower computation complexity is also devised, which can achieve near-optimal performance according to our simulation results. The performance of different policies under different system settings is compared and analyzed in numerical results to provide useful insights for practical system designs.
\end{abstract}
\section{Introduction}
Recently, researchers have shown enormous interest (see, e.g, \cite{kaul2012real,wang2019minimizing2,ceran2019average,wang2018skip,wang2019minimizing,kaul2012status,gu2019timely,sun2017update,costa2016age,gu2019minimizing,kadota2018optimizing,kadota2018scheduling,kadota2019minimizing,yates2017status,jiang2018can,maatouk2019minimizing1,chen2020age,hsu2019scheduling}) in a new performance metric, termed \textit{Age of Information} (AoI), thanks to its capability in characterizing the timeliness of data transmission in status update systems. The timeliness of status update is of great importance, especially in real-time monitoring applications, in which the dynamics of the monitored processes need to be well grasped at the monitor side for further actions. The AoI is defined as the time elapsed since the generation time of the latest received status update at the destination \cite{kaul2012real}. According to this definition, the AoI is jointly determined by the transmission interval and the transmission delay. 

Early work on the analysis and optimization of AoI in various networks has mainly focused on the simple single-source system model \cite{ceran2019average,wang2019minimizing2,wang2018skip,wang2019minimizing,kaul2012real,kaul2012status,gu2019timely,sun2017update,costa2016age,gu2019minimizing}. Recent efforts on AoI optimization pay more attention to the more general multi-source systems \cite{hsu2019scheduling,yates2017status,kadota2018optimizing,kadota2018scheduling,kadota2019minimizing,jiang2018can,maatouk2019minimizing1,chen2020age}. For systems with multiple sources, the AoI of each user depends on the transmission scheduling of all devices. In this line of research, the authors in \cite{kadota2018optimizing} considered a base station (BS) receiving status updates from multiple nodes with a \textit{generate-at-will} status arrival model in the uplink. A BS serving status updates to multiple nodes in the downlink with the randomly generated status update was investigated in \cite{kadota2019minimizing}. Both of them derived the lower bound of the weighted sum of the expected AoI of the considered network and compared the lower bound with that of various suboptimal scheduling policies, including Whittle index policy and max-weight policy, etc. The authors in \cite{jiang2018can} also considered systems with stochastic status update arrivals and derived the Whittle index policy in closed form. A decentralized policy was proposed in \cite{jiang2018can}, which was shown to achieve near-optimal performance. Another branch of this research line is to analyze and optimize the AoI of the networks with random access protocols. Particularly, the AoI performance of slotted ALOHA was investigated in \cite{yates2017status,chen2020age} and that of Carrier Sense Multiple Access (CSMA) was investigated in \cite{maatouk2019minimizing1}. 
%\textit{generate-at-will}\cite{yates2017status,kadota2018optimizing,kadota2018scheduling}, and stochastic arrival \cite{kadota2019minimizing,jiang2018can} are two common models used to model the generation of status update at the sources. ALOHA and scheduled access with feedback policies were compared in \cite{yates2017status} and the optimal transmission probability of ALOHA for different sources was further derived. The throughput constraint of the network was considered in \cite{kadota2018optimizing}. 

All aforementioned studies on AoI have concentrated on the orthogonal multiple access (OMA) scheme. That is, only one status update packet can be delivered and received in each time slot. Very recently, the authors in \cite{maatouk2019minimizing} have for the first time investigated the potential of applying non-orthogonal multiple access (NOMA) in reducing the average AoI of a two-node network. The results in \cite{maatouk2019minimizing} showed that OMA and NOMA can outperform each other in different setups. In fact, NOMA has been regarded as a promising technique to deal with large-scale Internet of Thing (IoT) deployment \cite{ding2017application,saito2013non,dong2017non,yu2017performance}. The basic idea of NOMA is to leverage the power domain to enable multiple clients to be served at the same time or frequency band. Compared to OMA, NOMA has the potential to reduce AoI by improving spectrum utilization efficiency. Specifically, more than one client can be served by the BS using NOMA, resulting in a possible AoI drop of more than one client. However, in OMA, only the served client may have AoI drop and the AoI of all other clients will increase. In this context, a natural question arises: how should a multiuser system adaptively switch between OMA and NOMA modes to minimize the long-term average weighted sum of AoI of the network? To the best of authors' knowledge, the answer to this question remains unknown in the literature. The NOMA scheme allows the BS to serve more clients in each time slot at the cost of a high transmission error probability, while the OMA scheme serves at most one client in each time slot with a smaller transmission error probability. This makes the optimal multiuser scheduling problem with adaptive NOMA/OMA non-trivial. In Fig.\ref{fig-0} we depict an example of the AoI evolution under the adopted adaptive NOMA/OMA scheduling for a two-client network. We can observe from Fig. \ref{fig-0} that the BS may take a risk to serve both clients in order to achieve small AoI for both clients at next time slot when the age difference between clients is relatively small. When the age difference between clients is large with one age being small, the BS tends to use OMA to serve the client with larger AoI.
\begin{figure}[!tbp]
	\centerline{\includegraphics[width=0.5\textwidth]{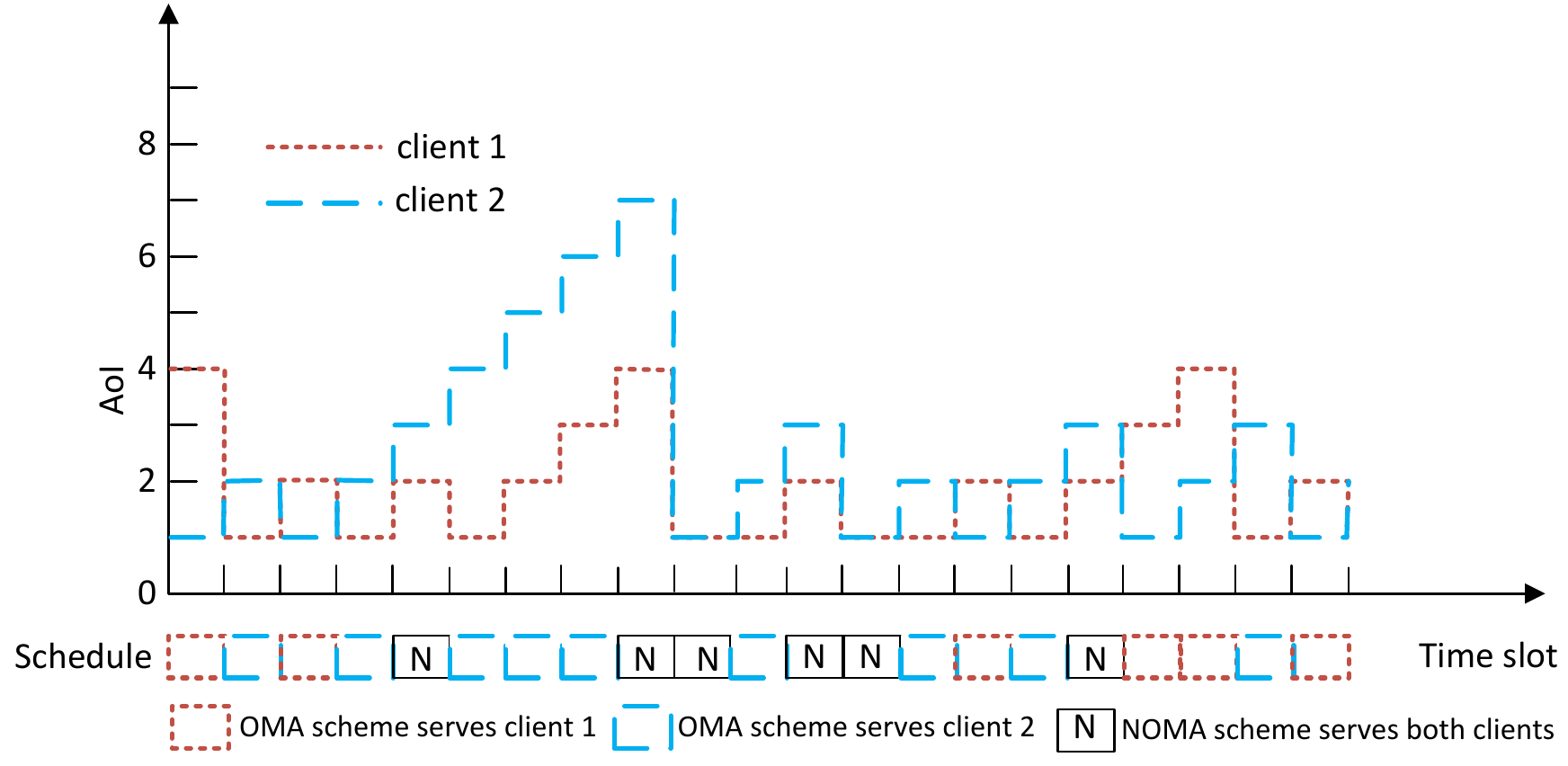}}
	\caption{An illustration of AoI evolution for a two-client network under the adopted adaptive NOMA/OMA scheduling.}
	\label{fig-0}
	%\vspace{-1em}
\end{figure}

Motivated by the gap above, in this paper we consider a wireless network with a BS that conducts timely status updates to multiple clients in a time-slotted manner. The BS is able to adaptively switch between NOMA and OMA for the downlink transmission. To achieve reduced AoI performance, the BS needs to decide which scheme (i.e., NOMA or OMA) to use at the beginning of each time slot. For the OMA scheme, the BS should further decide which client to serve. For the NOMA scheme, the BS needs to further decide the power allocated to each scheduled client. That is, when using NOMA, the BS should decide which clients to serve by allocating non-zero power for status update transmission to these clients; the rest unselected clients will be allocated with zero power.
\subsection{Contributions}
The main contributions of this paper lie in the following two aspects: 
 \begin{itemize}
 	\item For the two-client scenario,  we develop the optimal policy for the BS to decide whether to use NOMA or OMA for each downlink transmission based on the instantaneous AoI of both clients by formulating a Markov Decision Process (MDP) problem. We prove the existence of the optimal stationary and deterministic policy, and perform action elimination to reduce the action space for lower computation complexity. The optimal policy is shown to have a switching-type property with obvious decision switching boundaries. A suboptimal policy with lower computation complexity is also proposed, which can achieve near-optimal performance, as shown by the simulation results.
 	\item For the multi-client scenario, the optimal policy is not computationally tractable due to the exponentially increasing state space for linearly increasing number of clients, the coupled AoI evolution across clients and large action space considering different combinations of power allocated to each client. To adaptively switch between NOMA and OMA, we formulate a nonlinear optimization problem to determine the optimal power allocated to each client by maximizing the weighted sum of expected AoI drop of the network within each time slot, inspired by the near-optimal policy and the max-weight policy in \cite{kadota2018optimizing,kadota2018scheduling,kadota2019minimizing}. 
 	We manage to resolve the formulated problem by approximating it as a convex optimization problem. We also derive the upper bound of the gap between the approximate convex problem and the original nonlinear, nonconvex problem. Simulation results show the effectiveness of the adopted approximation. The performance of the adaptive NOMA/OMA scheme by solving the convex optimization  problem is shown to be close to that of max-weight policy solved by exhaustive search. Besides, the adaptive NOMA/OMA scheme can achieve significantly lower average AoI, comparing to OMA scheme, especially when the number of clients in the network is large and the transmission SNR is high.
 	%By analyzing the property of the expected AoI of the network on the power domain, we introduce a convex problem to approximate the nonlinear optimization problem and derive the upper bound of the gap between the introduced convex problem and original nonlinear problem. 
 \end{itemize}
 
% To minimize the AoI of the network, we develop an optimal policy for the BS to decide whether to use NOMA or OMA for each downlink transmission based on the instantaneous AoI of both clients by formulating a Markov Decision Process (MDP) problem. We prove the existence of the optimal stationary and deterministic policy, and perform action elimination to reduce the action space for lower computation complexity. The optimal policy is shown to have a switching-type property with obvious decision switching boundaries. A suboptimal policy with lower computation complexity is also proposed, which can achieve near-optimal performance according to our simulation results. The approximate average AoI performance of the suboptimal policy is also derived by applying a two-dimensional Markov chain. The performance of different policies under different system settings is compared and analyzed in numerical results to provide useful insights for practical system designs.
\subsection{Related Work}
%In practical systems such as wireless sensor networks, more than one device should update its measurement to the monitor timely. Thus, much of the current literature on AoI pays particular attention to scheduling the communication in a multi-user network for the sake of minimizing the average AoI of the network. 
We note that MDP method has been widely used in designing optimal scheduling policies for average AoI minimization \cite{ceran2019average,wang2019minimizing2,wang2018skip,wang2019minimizing,hsu2019scheduling}. In multiuser systems, the states of the system are jointly determined by the AoI values of all users, where the MDP method becomes intractable as the number of users increases. This is because the increasing number of users will lead to exponentially exploding state space and enormous computation complexity, known as the curse of dimensionality \cite{powell2007approximate}. Thus, several attempts \cite{hsu2019scheduling,kadota2018optimizing,kadota2018scheduling,kadota2019minimizing,jiang2018can,maatouk2019minimizing1} have been made to seek for low-complexity scheduling algorithms. Whittle index policy has been investigated in \cite{hsu2019scheduling,kadota2018optimizing,kadota2018scheduling,jiang2018can}, where the indexability of their considered problem was proved. This policy demonstrated near-optimal performance in numerical simulations. To implement the Whittle index policy, the Whittle index function needs to be derived beforehand and the user with the largest Whittle index will be scheduled to update its status. However, it can be challenging to prove indexability and derive closed-form Whittle index function for many problems\cite{gittins2011multi}. To address these issues, the authors in\cite{sun2019optimizing} proposed an Approximate Index Policy. On the other hand, the max-weight policy has been studied in \cite{kadota2018optimizing,kadota2018scheduling,kadota2019minimizing} and the upper bound of its average age performance was analyzed. Simulation results in \cite{kadota2018optimizing,kadota2018scheduling} showed negligible performance gap between the max-weight policy and the optimal policy, and similar performance between Whittle index policy and max-weight policy. 
All the aforementioned work focused on OMA scheme, i.e., different users cannot update their status simultaneously. The potentials of NOMA scheme on reducing AoI were first investigated in \cite{maatouk2019minimizing} considering a simple two-user network. The analytical expression of the total average AoI of the network using  NOMA scheme and that of conventional OMA environments were derived via Stochastic Hybrid Systems (SHS) and compared in different setups. The simulation results have illustrated the advantage of NOMA for the case of relatively high spectral efficiency in comparison with OMA. The authors in \cite{maatouk2019minimizing} focused on analyzing the AoI of two-user network that always uses NOMA to investigate the potential of NOMA scheme by comparing it with the AoI of same network adopting OMA scheme. In contrast, our work considers how to dynamically schedule the communications in a more general multiuser system by adaptively switching between OMA and NOMA modes to minimize the AoI of the network. The considered system is more practical due to the increased number of users and the scheduling scheme is more comprehensive including which user(s) to schedule and the corresponding power allocation. 

\subsection{Organization}
The rest of the paper is organized as follows. Section II introduces the system model. We study optimal policy for the two-user scenario and propose a near-optimal policy in Section III. Section IV studies the multi-client scenario. Numerical results are presented in Section V to validate the theoretical analysis and the effectiveness of the proposed adaptive NOMA/OMA scheme. Finally, conclusions are drawn in Section VI.

\section{System Model}
We consider a multiuser wireless network, in which a BS conducts timely status updates to $N$ clients in a slotted manner. At the beginning of each time slot, the BS can generate a status update packet for each client, which is known as \textit{generate-at-will} in the literature \cite{ceran2019average,kadota2018optimizing,wang2019minimizing2}. Adaptive NOMA/OMA transmission scheme is adopted by the BS. Specifically, the BS can adaptively switch between NOMA and OMA for the downlink transmission. With NOMA, it is possible for more than one client to receive their packets simultaneously within one time slot. At the end of each time slot, if client $i$ has received its packet successfully from the BS, it will send an acknowledgment (ACK) to the BS. The ACK link from all clients to the BS is considered to be error-free and delay-free. 

We use \textit{Age of Information} (AoI) \cite{kaul2012real} to characterize the timeliness of the information received at each client. AoI is defined as the time elapsed since the generation time of the latest received information at the destination side. Mathematically, the AoI of client $i$ in time $t$, denoted by $\Delta_i(t)$, is $t-u_i(t)$, where $u_i(t)$ denotes the generation time of latest received status update at time $t$. According to the considered \textit{generate-at-will} model, if client $i$ has successfully received its status update from the BS, its AoI will decrease to $1$, otherwise its AoI increases by $1$. Mathematically, we have
\begin{equation}\label{aoievo}
\Delta_i(t+1)=
\left\{
\begin{array}{rcl}
\Delta_i(t)+1 ,& &v_i(t)=0 ,\\
1,& &v_i(t)=1, \\
\end{array}
\right.
\end{equation} where $v_i(t)$ is the indicator that is equal to $1$ when the client $i$ receives its status update correctly from the BS in time slot $t$, and $v_i(t)=0$ otherwise. The weighted sum of the expected AoI of all clients is adopted to measure the network-wide information timeliness, which is given by
\begin{equation}\label{wa}
\bar{\boldmath \Delta}=\lim_{T \rightarrow \infty}\sup\frac{1}{T}\mathbb{E}\left[\sum_{i=1}^{N}\sum_{t=1}^{T}w_i\Delta_i(t)\right],
\end{equation} where $w_i$ is the weight coefficient of client $i$ with $\sum_{i=1}^N w_i=1$, and the expectation is taken over all possible system dynamics.
%\section{The case of two clients}

For ease of understanding, we first consider the two-client scenario, i.e., $N=2$. We later will extend our design to the general case with more clients. In the OMA mode, the BS only conducts transmission to a single client. In this context, if time slot $t$ is assigned for the transmission to client $i$, $i\in\{1,2\}$, the signal received at the client $i$ can be written as 
\begin{equation}\label{observe}
y_i^O(t)=h_i(t) \sqrt{P}s_i(t)+n_i(t),
\end{equation} where $P$ is the constant transmission power of the BS; $s_i$ is the status update message from the BS to client $i$; $h_i$ is the channel coefficient between the BS and client $i$. Specifically,
\begin{equation}\label{h1}
 h_i=\sqrt{d_i^{-\tau}}g_i,
 \end{equation} 
 where the normalized distance $d_i=c_i/c_0$, with $c_i$ and $c_0$ denoting the distance between client $i$ to the BS and the baseline distance, respectively. Parameter $\tau$ denotes the path loss exponent and $g_i \sim \mathcal{CN}(0,1)$ with $\mathcal{CN}$ denoting complex normal distribution. Without loss of generality, we consider $c_1<c_2$, i.e., $\mathbb{E}[{|h_1|}^2]>\mathbb{E}[{|h_2|}^2]$. Random variable $n_i$ is the complex additive Gaussian noise with variance $\sigma_i^2$. For simplicity, we assume the variance of $n_i$ is identical for both clients, i.e., $\sigma_i^2=\sigma^2$, $\forall i$. After receiving the signal, the information can be decoded in an interference-free manner with a SNR $\gamma_i={|h_i|}^2\rho$, where $\rho=P/\sigma^2$ is the transmission SNR. Then, the rate for client $i$ can be expressed as $R_i^{OMA}=\log(1+\gamma_i)$. The outage probability at client $i$ using OMA is given by
\begin{equation}\label{oma}
\begin{split}
P_i^{O}=1-P\left(R_i^{OMA}\geq R_i\right)=1-\exp\left(-\frac{(2^{R_i}-1)d_i^{\tau}}{\rho}\right),
\end{split}
\end{equation} where $R_i$ is the target rate of client $i$. For simplicity, we assume that the target rates of both clients are the same, i.e., $R_1 = R_2 = R$. Note that the framework developed for the two-user scenario, can be readily extended to the case with distinct target rates. %{\color{red}Without loss of generality, we assume the target rate of each client is same, i.e., $R_i=R$ and $R$ will be used in the following equations.}

On the other hand, when NOMA is conducted in time slot $t$, the signals to different clients are combined in the power domain at the BS by allocating different power levels to them. Through successive interference cancellation (SIC), it is possible for two clients to successfully recover their corresponding information in the same time slot. We consider fixed power transmission, and the observation at client $i$ can be expressed as 
\begin{equation}
y_i^N(t)=h_i(t) (\sqrt{\alpha_1 P}s_1(t)+\sqrt{\alpha_2 P}s_2(t))+n_i(t),
\end{equation} where $\alpha_i$ is the power allocation coefficient, and we readily have $\alpha_1+\alpha_2=1$ to achieve the best possible performance. It is assumed that the BS only has the knowledge of statistical channel state information (CSI) of its channels to both clients, while the clients as receivers have perfect knowledge of CSI, as in \cite{cui2016novel,yu2017performance}. Thus, we have $\alpha_1<\alpha_2$ according to the NOMA principle. 

Then, for client $2$ (i.e., the far user), it decodes its message from the BS directly by treating $s_1$ as interference. The received SINR can be written as
$\gamma_{22}=\alpha_2 {|h_2|}^2/(\alpha_1 {|h_2|}^2+1/\rho)$. Therefore, the outage probability of client $2$ using NOMA is given by
\begin{equation}
\label{outage2}
\begin{split}
P_2^{N}&=1-P(\log(1+\gamma_{22})\geq R)=1-\exp\left(-\frac{(2^{R}-1)d_2^\tau}{\rho (\alpha_2-\alpha_1(2^{R}-1))}\right),
\end{split}
\end{equation} where we enforce  $\alpha_2-\alpha_1(2^{R}-1)>0$, i.e., $\alpha_2>\frac{2^R-1}{2^R}$.

For client $1$ (i.e., the near user), it will conduct SIC. Specifically, client $1$ will first decode $s_2$ as what client $2$ has done by treating $s_1$ as interference. The received SINR of client $1$ when decoding $s_2$, denoted by $\gamma_{12}$, can thus be similarly expressed as
$\gamma_{12}=\alpha_2 {|h_1|}^2/(\alpha_1 {|h_1|}^2+1/\rho)$. Once $s_2$ is successfully decoded, client $1$ will then decode $s_1$ without interference, and the resultant SNR is
$\gamma_{11}=\alpha_1 {|h_1|}^2\rho$. The outage probability of client $1$ using NOMA can thus be calculated as
%\begin{small}
\begin{equation}
\label{outage1}
\begin{split}
&P_1^{N}=1-P(\log(1+\gamma_{12})\geq R\ \& \ \log(1+\gamma_{11})\geq R)\\
&=1-\exp\left(-\max\left\{\frac{(2^{R}-1)d_1^\tau}{\rho (\alpha_2-\alpha_1(2^{R}-1))},\frac{(2^{R}-1)d_1^{\tau}}{\rho\alpha_1}\right\}\right).
\end{split}
\end{equation}
%\end{small}

Comparing the above outage probability expressions between NOMA and OMA schemes, we can find that NOMA offers more chance for the BS to transmit fresh status updates to both clients at the cost of a higher outage probability. Thus, to maintain the freshness of the information received at each client, at the beginning of each time slot, the BS needs to carefully decide whether to use NOMA or OMA scheme.  In addition, the outage probability of NOMA is determined by the power allocation among the two clients. As such, when using NOMA, the BS should appropriately allocate power for the transmission to each client. The power allocated to each client is considered to be discrete in the two-client system. Specifically, the power allocated to client $i$, denoted by $p_i$, can only take the value from the discrete set $\{0,p,2p,3p,..Lp\}$ with $p=P/L$ and $p_1+p_2=P$, as $\alpha_1=1-\alpha_2$. That is, $\alpha_i$ can take the value from $\{0,\frac{1}{L},\frac{2}{L},\frac{3}{L},..,1\}$. As client $2$ is far from the BS (i.e., $c_1<c_2$), to effectively use NOMA, $\alpha_2$ should be larger than $\alpha_1$ when applying NOMA, i.e., $\alpha_2>0.5$. Combining it with the previous condition $\alpha_2>\frac{2^R-1}{2^R}$, one can deduce that $\alpha_2$ can only take value from $\{0,\max\{\frac{1}{2}+\frac{1}{L},\lceil\frac{(2^R-1)L}{2^R}\rceil\frac{1}{L}\},\max\{\frac{1}{2}+\frac{1}{L},\lceil\frac{(2^R-1)L}{2^R}\rceil\frac{1}{L}\}+\frac{1}{L},...,1\}$.

Let $\alpha_2(t)$ denote the power allocation coefficient for client $2$ in time slot $t$. Specifically, $\alpha_2(t)=0$ or $\alpha_2(t)=1$ indicates the BS uses OMA scheme, conducting orthogonal transmission to client $1$ and client $2$, respectively; otherwise, the BS uses NOMA scheme, serving both clients with the amount of power $\alpha_2(t)P $ allocated to client $2$ and $(1-\alpha_2(t))P$ to client $1$.

%Let $\alpha_i(t)$ denote whether the BS conducts transmission for client $i$ at time slot $t$. Specifically, $\mu_i(t)=1$ indicates that the BS conducts transmission to client $i$, $\mu_i(t)=0$ otherwise. If the BS chooses to adopt OMA in time slot $t$, the BS can conduct a transmission to only one client, i.e., $\sum_{i=1}^{2}\mu_i(t)=1$. If the BS decides to use the NOMA transmission scheme, it can conduct transmission to both clients, i.e., $\mu_i(t)=1\ , \forall i$.  

%{\color{red}Let $\pi$ denote the transmission policy at the BS, consisting of a sequence of actions at each time slot, denoting by $\{a_t\}$. $a_t\in \{0,\max\{\lceil{\frac{N}{2}\rceil}+1,\lceil\frac{(2^R-1)N}{2^R}\rceil\},...,N\}$ indicates that the BS allocates $a_tp$ amount of power to client $2$.}
Let $\pi$ denote the stationary transmission policy at the BS, which maps system states to action space. Denoting $a_t$ as the action at time slot $t$, $a_t\in \{0,\max\{\lceil{\frac{L}{2}\rceil}+1,\lceil\frac{(2^R-1)L}{2^R}\rceil\},...,L\}$ indicates that the BS allocates $a_tp$ amount of power to client $2$. If $a_t=0$, the BS chooses OMA scheme and only transmits information to client $1$; if $a_t=L$, the BS chooses OMA scheme and transmits information to client $2$;  otherwise, the BS chooses NOMA scheme, with $a_tp$ amount of power allocated to client $2$ and $P-a_tp$ allocated to client $1$. Our design objective is to find the optimal policy to be adopted by the BS that can adaptively switch between NOMA and OMA schemes to minimize the weighted sum of the expected AoI for both clients. The problem can be formally formulated as follows
\begin{prob}
\label{p1}
		\begin{equation}
	    \min_{\pi} \bar{\boldmath \Delta}(\pi).\\
		\end{equation}
\end{prob} 
\section{Optimal and Near-Optimal Policies for Two-Client System}
In this section, we resolve Problem \ref{p1} by formulating it as an MDP problem and investigate the age-optimal policy that minimizes the weighted sum of the expected AoI of both clients. By analyzing the structural results of the optimal policy, we then devise a near-optimal policy with lower computation complexity.
\subsection{MDP Formulation} 
We first recast Problem \ref{p1} into an MDP problem, described by a tuple $\{\mathcal{S},\mathcal{A},\mathrm{P},r\}$, where
 \begin{itemize}
	\item State space $\mathcal{S}=\mathbb{Z^+}\times\mathbb{Z^+}$: The state in time slot $t$ is composed by the instantaneous AoI of both clients, $s_t\triangleq (\Delta_{1,t},\Delta_{2,t})$.
	\item Action space $\mathcal{A}=\{0,\max\{\lceil{\frac{L}{2}\rceil}+1,\lceil\frac{(2^R-1)L}{2^R}\rceil\},...,L\}$: the detailed description of action $a_t \in \mathcal{A}$ has been provided at the end of the previous section. 
	\item Transition probabilities $\mathrm{P}$: $P(s_{t+1}|s_t,a_t)$ is the probability of the transition from state $s_t$ to $s_{t+1}$ when taking action $a_t$. According to the outage probability of both clients using either NOMA or OMA given in Section II, we have the following transition probabilities,
	\begin{equation}
	\label{e3}
	\begin{aligned}
	&P((1,\Delta_2+1)|(\Delta_1,\Delta_2),a=0)=1-P_1^O,\\
	&P((\Delta_1+1,\Delta_2+1)|(\Delta_1,\Delta_2),a=0)=P_1^O,\\
	&P((\Delta_1+1,1)|(\Delta_1,\Delta_2),a=L)=1-P_2^O,\\
	&P((\Delta_1+1,\Delta_2+1)|(\Delta_1,\Delta_2),a=L)=P_2^O,\\	
	\end{aligned}
	\end{equation} and for $i\neq 0,N$
	\begin{equation}
	\begin{aligned}
	\label{e3*}
	&P((1,\Delta_2+1)|(\Delta_1,\Delta_2),a=i)=(1-P_1^N(a))P_2^N(a),\\
	&P((\Delta_1+1,1)|(\Delta_1,\Delta_2),a=i)=(1-P_2^N(a))P_1^N(a),\\	
	&P((1,1)|(\Delta_1,\Delta_2),a=i)=(1-P_1^N(a))(1-P_2^N(a)),\\
	&P((\Delta_1+1,\Delta_2+1)|(\Delta_1,\Delta_2),a=i)=P_1^N(a)P_2^N(a),
	\end{aligned}
	\end{equation}where $P_1^N(a)$ and $P_2^N(a)$ are the outage probability of client $1$ and client $2$, respectively, using NOMA with $\alpha_1=1-\frac{a}{L}$ and $\alpha_2=\frac{a}{L}$. Note that in \eqref{e3} and \eqref{e3*}, the time superscript for the state $(\Delta_{1,t},\Delta_{2,t})$ and action $a_t$ is omitted for brevity.
	\item  $r: \mathcal{S} \times \mathcal{A}  \rightarrow \mathbb{R}$ is the one-stage reward function of state-action pairs, defined as $r(s_t,a_t)=w_1\Delta_{1,t}+w_2\Delta_{2,t}$.
 \end{itemize}
Given any initial state $s_0$, the infinite-horizon average reward of any feasible policy $\pi$, can be expressed as 
\begin{equation}
\label{e4}
C(\pi,s_0)=\lim_{T \rightarrow \infty}\sup\frac{1}{T} \sum_{k=0}^{T}{\mathbb E}_{s_0}^\pi[r(s_k,a_k)|s_0].
\end{equation} We are now ready to transform Problem \ref{p1} to the following MDP problem
\begin{prob}
	\label{p2}
	\begin{equation}
	\label{pro2}
	\min_{\pi} C(\pi,s_0).
	\end{equation}
\end{prob} 
To proceed, we first investigate the existence of an optimal stationary and deterministic policy of Problem \ref{p2} and arrive at the following theorem.
\begin{theorem}
	\label{TE}
	There exists a constant $J^{*}$, a bounded function $h(\Delta_1,\Delta_2):\mathcal{S} \rightarrow \mathbb{R}$ and a stationary and deterministic policy $\pi^{*}$, satisfying the average reward optimality equation,
	\begin{equation}
	\label{e10}
	J^{*}+h(\Delta_1,\Delta_2)=\min_{a\in \mathcal{A}} (w_1\Delta_1+w_2\Delta_2 \mathbb  +{\mathbb{E}}[h(\hat{\Delta}_1,\hat{\Delta}_2)]),
	\end{equation} $\forall (\Delta_1,\Delta_2) \in \mathcal{S}$, where $\pi^{*}$ is the optimal policy, $J^{*}$ is the optimal average reward, and $(\hat{\Delta}_1,\hat{\Delta}_2)$ is the next state after $(\Delta_1,\Delta_2)$ taking action $a$.  
\end{theorem}
\begin{proof}
	See Appendix \ref{A1}.
\end{proof}	
According to Theorem \ref{TE}, the optimal policy is
stationary and deterministic, i.e., it is time-invariant and deterministically selects an action in each time slot with no randomization.
\subsection{Action Elimination}
In this subsection, we establish action elimination by analyzing the property of the formulated MDP problem, which can reduce action space of each state for lower computation complexity. According to \eqref{outage2} and \eqref{outage1}, and the fact $\alpha_1+\alpha_2=1$, the outage probability of client $2$ using NOMA (i.e., $P_2^N$) is decreasing in $\alpha_2$, i.e., $P_2^N(a)$ is decreasing in action $a$ when $ \max\{\lceil{\frac{L}{2}\rceil}+1,\lceil\frac{(2^R-1)L}{2^R}\rceil\}<a<L$. However, the outage probability of client $1$ using NOMA (i.e., $P_1^N$) is decreasing in $\alpha_2$ when $\frac{2^R-1}{2^R}<\alpha_2<\frac{2^R}{2^R+1}$ and is increasing in $\alpha_2$ when $\frac{2^R}{2^R+1}<\alpha_2<1$. That is, $P_1^N(a)$ is decreasing in $a$ when $a\in \{\max\{\lceil{\frac{L}{2}\rceil}+1,\lceil\frac{(2^R-1)L}{2^R}\rceil\},...,\lfloor\frac{2^RL}{2^R+1}\rfloor\}$ and increasing in $a$ when $a\in\{\lceil\frac{2^RL}{2^R+1}\rceil,\lceil\frac{2^RL}{2^R+1}\rceil+1,..., L-1\}$. As such, the action  $a=\lfloor\frac{2^RL}{2^R+1}\rfloor$  has a better performance in reducing AoI of both clients, with lower outage probability comparing to $a\in \{\max\{\lceil{\frac{L}{2}\rceil}+1,\lceil\frac{(2^R-1)L}{2^R}\rceil\},\max\{\lceil{\frac{L}{2}\rceil}+1,\lceil\frac{(2^R-1)L}{2^R}\rceil\}+1,...,\lfloor\frac{2^RL}{2^R+1}\rfloor\}$. Thus, the action space can be reduced to $a\in\{0,\lfloor\frac{2^RL}{2^R+1}\rfloor,\lfloor\frac{2^RL}{2^R+1}\rfloor+1,..., L\}$.
\subsection{Structural Results on Optimal Policy}
In this subsection, we derive two structural results of the optimal policy that offer an effective way to reduce the offline computation complexity and online implementation hardware requirement.
\begin{theorem}
	\label{T2}
	The optimal policy $\pi^*$ has a switching-type policy. That is, denoting $c$ and $d$ as any action from action space $\{0,\lfloor\frac{2^RL}{2^R+1}\rfloor,\lfloor\frac{2^RL}{2^R+1}\rfloor+1,..., L\}$, 
	\begin{itemize}
		\item If $\pi^*((\Delta_1,\Delta_2))=c$, then $\pi^*((\Delta_1,\Delta_2+z))=d$, where $z$ is any positive integer and $d\geq c$,
		\item If $\pi^*((\Delta_1,\Delta_2))=c$, then $\pi^*((\Delta_1+z,\Delta_2))=d$, where $z$ is any positive integer and $d\leq c$.
	\end{itemize}
\end{theorem}
\begin{proof}
	See Appendix \ref{A2}.
\end{proof}	

Given the structure of the optimal policy, only the decision switching boundary is needed for implementation, rather than storing each state-action pair in the optimal policy, which significantly reduces the memory for the hardware. In addition, based on the structure, a special algorithm can be developed accordingly as in \cite[Althorithm 1]{wang2018skip} to reduce the complexity in calculating the optimal policy. 

\subsection{Near-optimal Policy}
In this subsection, we propose a near-optimal policy with lower computation complexity comparing with that of the optimal MDP policy. %Inspired by the max-weight policy in \cite{kadota2018scheduling}, the proposed suboptimal policy makes use of the transition probability in the MDP and only minimizes the expectation of the reward of the next stage. According to \eqref{e3}, given current state $s=(\Delta_1,\Delta_2)$, the expected reward of next stage $\hat{s}$ can be expressed as 
%\begin{equation}\label{ew}
%\mathbb{E}[r(\hat{s},a)]=\left\{
%\begin{array}{rcl}
%1+w_1P_1^O\Delta_1+w_2\Delta_2,& &\text{if }  a=1;\\
%1+w_1\Delta_1+w_2P_2^O\Delta_2,& &\text{if } a=N\\
%1+w_1P_1^N(a)\Delta_1+w_2P_2^N(a)\Delta_2,& &\text{otherwise.}  \\
%\end{array}
%\right.
%\end{equation} Then, the action of state $s$ in the proposed suboptimal policy $\bar{\pi}$ can be given by
%\begin{equation}
%\bar{\pi}(s)=\arg \min_a \mathbb{E}[r(\hat{s},a)].
%\end{equation} 
Inspired by the max-weight policy in \cite{kadota2018scheduling}, the proposed suboptimal policy makes use of the transition probability of the underlying MDP and only maximizes the weighted sum of the expected AoI drop within each time slot, i.e., the weighted sum of the expected difference between the age of current state and the possible age of next state. According to \eqref{e3}, given the current state $s=(\Delta_1,\Delta_2)$, the expected AoI drop, denoted by $\mathbb{E}[\eta(s,a)]$, can be expressed as 
	\begin{equation}\label{ew1}
	\mathbb{E}[\eta(s,a)]=\left\{
	\begin{array}{rccl}
&	w_1(1-P_1^O)\Delta_1-1,& &\text{if }  a=1;\\
&	w_2(1-P_2^O)\Delta_2-1,& &\text{if } a=L\\
&	w_1(1-P_1^N(a))\Delta_1+w_2(1-P_2^N(a))\Delta_2-1,& &\text{otherwise.}  \\
	\end{array}
	\right.
	\end{equation} Then, the action of state $s$ in the proposed suboptimal policy $\bar{\pi}$ can be given by
	\begin{equation}
	\label{mw}
	\bar{\pi}(s)=\arg \max_a \mathbb{E}[\eta(s,a)].
	\end{equation} 
The suboptimal policy is simple and easy to implement. Moreover,  as we show via the numerical results in Section IV, the suboptimal policy can achieve near-optimal performance. In addition, the suboptimal policy can be readily extended to continuous power scenario, i.e., in each time slot, finding the optimal power allocated to each client to maximize the weighted sum of the expected AoI drop where $P_1^N(a)$ and $P_2^N(a)$ in \eqref{ew1} are replaced by the outage probability of each client using NOMA with continuous power allocated to client $2$.  %The structure of the suboptimal policy offers some insights when choosing NOMA for transmission. That is, if $\mathbb{E}[r(\hat{s},i)]>\mathbb{E}[r(\hat{s},a)]$, $\forall a \in\{1,3\}$ and $\forall s\in \mathcal{S}$, then using NOMA is meaningless.

 \section{Extension to multiple clients $N>2$}
 Recall that the BS aims to deliver status updates to all clients in a timely manner. To that end, the BS needs to carefully decide the transmission power allocated to each client as explained in Section III. However, since state-space explodes exponentially as the number of clients and the power discretization levels increase, the MDP method elaborated in Section III is no longer computationally tractable for the multi-client scenarios.
 
 In this section, we extend our near-optimal policy proposed in Section III.D to the general case with a BS delivering timely status updates to $N$ clients ($N>2$) in a slotted manner using adaptive NOMA/OMA principle. At the beginning of each time slot, the BS needs to schedule transmission to clients. That is, the BS decides to transmit to which client(s) and allocates the transmission power to them. At the end of each time slot, if client $i$ has received its packet successfully from the BS, it will send an ACK to the BS. The observation at the $i$th client in time slot $t$ is given by 
 \begin{equation}
 y_i(t)=h_i(t)\sum_{j=1}^{N} \sqrt{p_j(t)}s_j(t)+n_i(t),
 \end{equation} where $s_j$ denotes the message from BS to client $j$ and $h_i$ denotes the channel coefficient between the BS and client $i$ as in \eqref{h1}. Without loss of generality, we consider the sorted distance $c_1>c_2>...>c_N$, i.e., $\mathbb{E}[{|h_1|}^2]<\mathbb{E}[{|h_2|}^2]<...<\mathbb{E}[{|h_N|}^2]$. Variable $p_j$ is the transmission power allocated to the message intended to client $j$ which satisfies the power limit $\bar{p}$, i.e., $\sum_{i=1}^{N}p_i\leq \bar{p}$, and $n_i \sim \mathcal{CN}(0,\sigma_i^2)$ is the complex additive Gaussian noise at client $i$. For simplicity, we assume the variance of $n_i$ is identical for all clients, i.e., $\sigma_i^2=\sigma^2$, $\forall i$.
 
 Denoted by $\mathcal{N}$ the set of all clients in the system, i.e., $\mathcal{N}=\{1,2,...,N\}$. Any subset $\mathcal{K}\subseteq\mathcal{N}$ denotes the possible set of clients to be served in each time slot. According to the NOMA principle, in the subset of clients selected to be served, a client with a smaller distance is assigned with a larger decoding order index \cite{xu2016outage,xu2017optimal}. Each selected client employs the successive interference cancellation (SIC) technique to decode the messages for clients with a smaller decoding order index in the selected client set first, and to remove the inter-user interference if the decoding is correct. Denote $\lambda_i$ as the indicator that equals $1$ when client $i$ is selected to transmit, and equals $0$ otherwise. Thus, if $K$ clients are selected to be served, then $\sum_{i=1}^{N}\lambda_i=K$. Let $m(k)$ denote the original client index among the $K$ selected clients whose decoding order is $k$, i.e., $\lambda_{m(k)}=1$, $k\in\{1,2,...,K\}$, $\forall k$. $m(.)$ is a single mapping that maps the set $\{1,2,...,K\}$ to the set $\{1,2,...,N\}$ where $K\leq N$. The sequence ${\{m(k)\}}_{k=1,2,...,K}$ consists of the set of clients selected for receiving status updates. Besides, according to the decoding order of NOMA, we have $m(1)<m(2)<...<m(K)$.
 
 Given the set of clients ${\{m(k)\}}_{k=1,2,...,K}$ to be served, denote by $R_{m(i)}^{m(j)}$ the rate for client $m(j)$ to detect client $m(i)$'s message. We consider $j\geq i$, indicating $m(j)\geq m(i)$. To correctly detect client $m(i)$'s message, client $m(j)$ should first successfully remove the interference from clients in ${\{m(k)\}}_{k=1,2,...,K}$ whose decoding order index is smaller than $m(i)$. Thus, the expression of $R_{m(i)}^{m(j)}$ is given by \cite{xu2016outage,xu2017optimal,cui2016novel}
 \begin{equation}
 R_{m(i)}^{m(j)}=\log\left(1+\frac{{|h_{m(j)}|}^2p_{m(j)}}{\sum_{k=i+1}^{K}{|h_{m(k)}|}^2p_{m(k)}+\sigma^2}\right).
 \end{equation} As the BS does not have perfect knowledge of CSI, outage may occur in the considered system. We define that if client $m(j)$ cannot detect its own message or the message of client $m(i)$ with smaller decoding index $m(j) \geq m(i)$ in the selected client set, then outage occurs at client $m(j)$ \cite{cui2016novel,ding2014performance}. Assume that the BS transmits one message to each client with
 the same fixed target rate $R$, the outage probability of client $m(j)$ can be expressed as \cite{xu2016outage,xu2017optimal}
 \begin{equation}
 \label{eqo1}
 \begin{aligned}
 P_{m(j)}^o&=1-P\left(R_{m(1)}^{m(j)}\geq R,...,R_{m(j)}^{m(j)}\geq R\right)\\
 &=1-\exp\left(-d_{m(j)}^\tau \max_{k=1,2,...,j}\left\{\frac{(2^{R}-1)\sigma^2}{p_{m(k)}-(2^{R}-1)\sum_{i=k+1}^{K}p_{m(i)}}\right\}\right).
 \end{aligned}
 \end{equation} We can see from \eqref{eqo1} that if $p_{m(k)}-(2^{R}-1)\sum_{i=k+1}^{N}p_{m(i)}\leq 0$, the outage probability of client $m(j)$ will always be $1$. Thus, for any client $m(k)$ selected to be served, i.e., $p_m(k)\neq 0$, the following condition needs to be satisfies 
 \begin{equation}
 \label{eqc1}
 p_{m(k)}> (2^{R}-1)\sum_{i=k+1}^{K}p_{m(i)}.
 \end{equation} Otherwise, an outage always occurs and the allocated power will be wasted. Moreover, if client $i$ is not served, i.e., $i\notin{\{m(k)\}}_{k=1,2,...,K}$ and $p_i=0$, its outage probability is $1$, otherwise, its outage probability will be smaller than $1$. Mathematically, we have
 \begin{equation}\label{outage}
\mathbb{E}[v_i(t)=1]=
\left\{
\begin{array}{rcl}
0,& & i\notin{\{m(k)\}}_{k=1,2,...,K},\\
1-P_{i}^o,& &i\in{\{m(k)\}}_{k=1,2,...,K}. \\
\end{array}
\right.
 \end{equation} Recall the $v_i(t)$ is the indicator that equals $1$ when client $i$ successfully receives its status update from the BS in time slot $t$. Let $\mathbf{p}(t)=\{p_1(t),p_2(t),...,p_N(t)\}$ denote the amount of transmission power allocated to each client satisfying $\sum_{i=1}^{N}p_i(t)\leq \bar{p}$. Give ${\{m(k)\}}_{k=1,2,...,K}$, we have $\sum_{i=1}^{K}p_{m(i)}(t)\leq \bar{p}$ and $p_i(t)=0$, $\forall i\notin {\{m(k)\}}_{k=1,2,...,K}$.
 
 Note that the special case $K=1$ indicates only one client will be served, i.e., client $m(1)$ will be served using OMA scheme. The corresponding outage probability becomes $1-\exp\left(-d_{m(1)}^\tau\frac{(2^{R}-1)\sigma^2}{p_{m(1)}}\right)$ as in \eqref{oma}.
 We now extend our near-optimal policy (i.e., problem in \eqref{mw}) to the multiple-client scenario by formulating the following power allocation problem.
 \begin{prob}
 	\label{p1*}
 	\begin{equation}
 	\begin{aligned}
 	&	\max_{\mathbf{p}(t)} \sum_{i=1}^{N}\left(1-P_i^o(\mathbf{p}(t))\right)w_i\Delta_i(t)\\
 	&\text{\rm s.t.}, \ \eqref{eqc1}, \ \sum_{i=1}^{N}p_i(t)\leq \bar{p},  \ p_i(t)\geq 0.
 	\end{aligned}
 	\end{equation}
 \end{prob} We note that in the above optimization problem, the instantaneous AoI of all clients in time slot $t$ will affect the power allocated to each client. Clients with smaller AoI are less likely to be served as the resultant AoI drop is insignificant.
 \subsection{Effective power allocation}
 In this subsection, we solve Problem \ref{p1*} to obtain the effective power allocation to minimize the weighted sum of expected AoI in two steps: 1) Step 1: design an optimal power allocation scheme to serve a fixed number of clients. That is, given $K$, find optimal ${\{m(k)\}}_{k=1,2,...,K}$ and $\left(p_{m(1)},p_{m(2)},...,p_{m(K)}\right)$; 2) Step 2: choose optimal $K\in\{1,2,..., N\}$ that achieves the maximum objective value. The detailed description of these two steps is given in the following.
 \subsubsection{\textbf{Step 1: Optimal power allocation to conduct transmission to fixed $K$ number of clients}}
 Given $K$, i.e., the number of clients to serve, the BS should decide which group of clients to serve, i.e., ${\{m(k)\}}_{k=1,2,...,K}$, and the power allocated to them, i.e., $\left(p_{m(1)}, p_{m(2)},..., p_{m(K)}\right)$. Recall that the power allocated to the unselected clients is $0$.
 
% Let $m(k)$ denote original client index among the $K$ selected client whose decoding order is $k$, i.e., $s_{m(k)}=1$, $k\in\{1,2,...,K\}$. $m(.)$ is single mapping, that maps set $\{1,2,...,K\}$ to set $\{1,2,...,N\}$, $K\leq N$. The sequence ${\{m(k)\}}_{k=1,2,...,K}$ consists the set of node chosen for transmission. Besides, due to the decoding order of NOMA, we have $m(1)<m(2)<...<m(K)$.} Thus, Eq. \eqref{eqc1} can be transferred into the following inequality
% \begin{equation}
% \label{eqc2}
% p_{m(k)}>(2^{\hat{R}}-1)\sum_{i=k+1}^{K}p_{m(i)}.
% \end{equation}  
 As in \cite[Eq.(15)]{xu2016outage}, we convert the power constraint described in \eqref{eqc1} to the following format to facilitate the use of power constraint,
 \begin{equation}
 \label{eqc3}
 \sum_{k=1}^{K}\hat{p}_{m(k)}{(r+1)}^{(k-1)}\leq \bar{p},
 \end{equation} where $r=2^{R}-1$ and 
 \begin{equation}\label{phat}
\hat{p}_{m(k)}=p_{m(k)}-r\sum_{i=k+1}^{K}p_{m(i)}.
 \end{equation}
 The outage probability of the selected client $m(k)$ can be expressed as
 \begin{equation}
 \label{eqo2}
 \begin{aligned}
 P_{m(k)}^o&=1-P\left(R_{m(1)}^{m(k)}\geq R,...,R_{m(k)}^{m(k)}\geq R\right)\\
 &=1-\exp\left(-d_{m(k)}^\tau r\sigma^2 \max_{t=1,2,...,k}\left\{\frac{1}{\hat{p}_{m(t)}}\right\}\right), \ k\in\{1,2,...,K\}.
 \end{aligned}
 \end{equation} For other unselected nodes, their outage probability is always equal to $1$. Recall that $c_1>c_2>...>c_N$, indicating $d_1^\tau>d_2^\tau>...>d_N^\tau$. We thus have $d_{m(1)}^\tau>d_{m(2)}^\tau>...>d_{m(K)}^\tau$. Note that only the selected clients may have AoI drop and the AoI of unselected clients will increase by one, and therefore the one-step weighted sum of expected AoI drop of the network is actually that of those selected clients. Hence, for a given $K$, Problem \ref{p1*} can be re-written as
 \begin{prob}
 	\label{p2*}
 	\begin{equation}
 	\begin{aligned}
 	&	\max_{\mathbf{p}(t)} \sum_{k=1}^{K}\left(1-P_{m(k)}^o(\mathbf{p}(t))\right)w_{m(k)}\Delta_{m(k)}(t)\\
 	\text{\rm s.t.,} &\ \eqref{eqc1}, \ \sum_{k=1}^{K}p_{m(k)}(t)\leq \bar{p},  \ p_i(t)= 0, \ \forall i \notin \{m(k)\}_{k=\{1,2,...,K\}}.
 	\end{aligned}
 	\end{equation}
 \end{prob} To further simplify the above problem, the variable transformation according to \eqref{eqc3} is applied, and Problem in \ref{p2*} can be transformed into the following equivalent form:
 \begin{prob}
 	\label{p2**}
 	\begin{equation}
 	\begin{aligned}
 	&	\max_{\mathbf{\hat{p}}(t),{\{m(k)\}}} \sum_{k=1}^{K}\left(1-P_{m(k)}^o(\mathbf{\hat{p}}(t))\right)w_{m(k)}\Delta_{m(k)}(t)\\
 	\text{\rm s.t.,} &\ \eqref{eqc3}, \ \mathbf{\hat{p}}(t)=(\hat{p}_{m(1)},\hat{p}_{m(2)},...,\hat{p}_{m(K)}), \ \hat{p}_{m(k)}> 0, \ \forall k\in\{1,2,...,K\}. 
 	\end{aligned}
 	\end{equation}
 \end{prob} This problem consists of two parts: 1) select which $K$ clients to serve, i.e., $\{m(k)\}$; 2) transferred power variable of these $K$ clients, i.e., $\mathbf{\hat{p}}(t)$, given  $\mathbf{\Delta}(t)=\{\Delta_1(t),\Delta_2(t),..., \Delta_N(t)\}$, $d_1^\tau>d_2^\tau>...>d_N^\tau$, $r$ and $\sigma^2$.
 
 %{\color{blue} $m(k)$ denotes the original index of the $k$th decoded client among the $K$ selected clients. The set of selected clients is denoted by $\mathcal{S}_1={\{m(k)\}}_{k=1,2,...,K}$ while $\mathcal{S}=\{1,2,....,N\}$ denotes the total clients. The clients in set $\mathcal{S}\setminus\mathcal{S}_1$ are not selected for transmission, i.e., $p_i=0$, $\forall i \in\mathcal{S}\setminus\mathcal{S}_1$.}
 
 %\textbf{Idea 1} \textit{Directly solve Problem 3, find $\{m(k)\}$ and $\{p_{m(1)}, p_{m(2)}, ..., p_{m(K)}\}$.}
 
% \textbf{Idea 2} \textit{Given $\{m(k)\}$, find $\{p_{m(1)}, p_{m(2)}, ..., p_{m(K)}\}$ (Problem 5-7) then analyze the relationship between $\{m(k)\}$ and its optimal performance under $\{p_{m(1)}, p_{m(2)}, ..., p_{m(K)}\}$ to calculate optimal $\{m(k)'\}$.}

 %\subsubsection{select $K$ clients to served, i.e., $\{m(k)\}$}
 
 %\subsubsection{power allocation for the $K$ clients, i.e., $\{p_{m(1)}, p_{m(2)}, ..., p_{m(K)}\}$}
 
 Suppose ${\{m(k)\}}_{k=1,2,...K}$ is known, we then solve Problem 5 as following (note that the time index $t$ is dropped hereafter for notation simplicity):
 \begin{prob}
 	\label{p6}
 \begin{subequations}
 	\begin{align}
 \quad  & \max_{\mathbf{\hat p}} ~\sum_{k=1}^K \exp\left(-d_{m(k)}^\tau r \sigma^2 \max\left\{\frac{1}{\hat p_{m(1)}}, \dots, \frac{1}{\hat p_{m(k)}}\right\}\right) w_{m(k)}\Delta_{m(k)} \label{eq:objective}
 	\\ 
 	&\text{s.t.,}  \quad
 	\sum_{k=1}^{K} (r+1)^{k-1} \hat p_{m(k)} \leq \overline p, \label{eq:constraint:power}
 	\\
 	&\quad \quad 
 	\hat p_{m(k)} >0, \quad k =1,\dots,K. \label{eq:constraint:positive}
 	\end{align}
 \end{subequations}
\end{prob}
In solving Problem \ref{p6}, we first have the following lemma.
 \begin{lemma}\label{lema1}
 	Adding the following constraint:
 	\begin{eqnarray}
 	\hat p_{m(1)} \geq \hat p_{m(2)} \geq \dots \geq \hat p_{m(K)} \nonumber
 	\end{eqnarray}
 	to Problem 6 will not change its optimal objective value of \eqref{eq:objective}
 \end{lemma}
\begin{proof}
	See the Appendix \ref{A3}.
\end{proof}
 By Lemma 1, we focus on solving the following problem to the same objective value as Problem 6, which can be solved in a simple and tractable way.
 \begin{prob}
 	\label{p7}
 \begin{subequations}
 	\begin{align}
 	\quad  & \max_{\mathbf{\hat p}} ~\sum_{k=1}^K w_{m(k)}\Delta_{m(k)}  \exp\left(- \frac{d_{m(k)}^\tau r \sigma^2}{\hat p_{m(k)}} \right) \label{eq:objective-1}
 	\\ 
 	&\text{s.t.}  \quad
 	\sum_{k=1}^{K} (r+1)^{k-1} \hat p_{m(k)} \leq \overline p \label{eq:constraint:power-1}
 	\\
 	& \quad\quad \hat p_{m(1)} \geq \hat p_{m(2)} \geq \dots \geq \hat p_{m(K)}   \label{eq:constraint:order-1}
 	\\
 	&\quad \quad 
 	\hat p_{m(k)} >0, \quad k =1,\dots,K \label{eq:constraint:positive-1}
 	\end{align}
 \end{subequations}
\end{prob}
 %Define variables $\tilde p_k := \frac{\hat p_{m(k)}}{d_{m(k)}^\tau r \sigma^2}$ and constants $\Gamma_k := (r+1)^{k-1} d_{m(k)}^\tau r \sigma^2$ and $\alpha_k := \ln{\Delta_{m(k)}}$ (note that $\alpha_k \geq 0$ because $\Delta_{m(k)} \geq 1$ for $k=1,\dots, K$). Then Problem 6 can be equivalently transformed to:
 % \begin{subequations}
 %\begin{align}
 %\text{\bf{Problem 7:}}\quad  & \max_{\tilde p} ~\sum_{k=1}^K \exp\left(- \frac{1}{\tilde p_k} + \alpha_k \right) \label{eq:objective-2}
 %\\ 
 %&\text{s.t.}  \quad
 % \sum_{k=1}^{K} \Gamma_k \tilde p_{k} \leq \overline p \label{eq:constraint:power-2}
 %\\
 %& \quad\quad d_{m(1)}^\tau \tilde p_{1} \geq d_{m(2)}^\tau \tilde p_2 \geq \dots \geq d_{m(K)}^\tau \tilde p_K   \label{eq:constraint:order-2}
 %\\
 %&\quad \quad 
 %\tilde p_{k} >0, \quad k =1,\dots,K \label{eq:constraint:positive-2}
 %\end{align}
 %\end{subequations}
 To proceed, we first investigate the properties of the objective function \eqref{eq:objective-1} in Problem \ref{p7}. We define
 \begin{eqnarray}
 G_k(\hat p_{m(k)}):=w_{m(k)}\Delta_{m(k)} \exp\left(- \frac{d_{m(k)}^\tau r \sigma^2}{\hat p_{m(k)}} \right).  \nonumber
 \end{eqnarray}
 The following properties hold for functions $G_k(\cdot)$, $k=1,\dots,K$:
 \begin{itemize}
 	\item $\lim_{\hat p_{m(k)}\rightarrow 0^+} G_k(\hat p_{m(k)}) = 0$. For convenience, we define $G_k(0) =0$;
 	\item $\lim_{\hat p_{m(k)}\rightarrow +\infty} G_k(\hat p_{m(k)}) = w_{m(k)}\Delta_{m(k)}$;
 	\item $G_k(\cdot)$ is strictly monotonically increasing on $(0,+\infty)$, which can be verified by checking $G_k'(\cdot)$;
 	\item $G_k(\cdot)$ is strictly convex on $[0,\frac{d_{m(k)}^\tau r \sigma^2}{2})$, and strictly concave on $[\frac{d_{m(k)}^\tau r \sigma^2}{2}, +\infty)$, which can be verified by checking $G_k''(\cdot)$.
 \end{itemize}
 
 Inspired by the properties above, we propose a convex upper approximation of $G_k(\cdot)$ as follows. We find a constant $\tilde p_{m(k)} > 0$ for each $k=1,\dots,K$, and replace the segment of $G_k(\cdot)$ on $[0,\tilde p_{m(k)}]$ by the straight line segment connecting two points $(0, G_k(0)=0)$ and $(\tilde p_{m(k)}, G_k(\tilde p_{m(k)}))$. At the same time, the straight line segment is a tangent line to $G_k(\cdot)$ at the point $(\tilde p_{m(k)}, G_k(\tilde p_{m(k)}))$. Therefore $\tilde p_{m(k)}$ can be calculated as follows:
 \begin{eqnarray}
 \frac{G_k(\tilde p_{m(k)}) - G_k(0)}{ \tilde p_{m(k)} - 0} = G_k'(\tilde p_{m(k)}) \nonumber
 \end{eqnarray}
 which leads to the result $\tilde p_{m(k)} = d_{m(k)}^\tau r \sigma^2$. Hence a convex upper approximate of $G_k(\cdot)$ is:
 \begin{eqnarray}
 \tilde G_k(\hat p_{m(k)}):=
 \begin{cases}
 \frac{w_{m(k)}\Delta_{m(k)} e^{-1}}{d_{m(k)}^\tau r \sigma^2}  \hat p_{m(k)}, \quad\quad\quad  0\leq \hat p_{m(k)} < d_{m(k)}^\tau r \sigma^2
 \nonumber \\
 w_{m(k)}\Delta_{m(k)} \exp\left(- \frac{d_{m(k)}^\tau r \sigma^2}{\hat p_{m(k)}} \right), \quad \hat p_{m(k)} \geq d_{m(k)}^\tau r \sigma^2\nonumber
 \end{cases}
 \end{eqnarray} 
  \begin{figure}[!tbp]
 	\centerline{\includegraphics[width=0.35\textwidth]{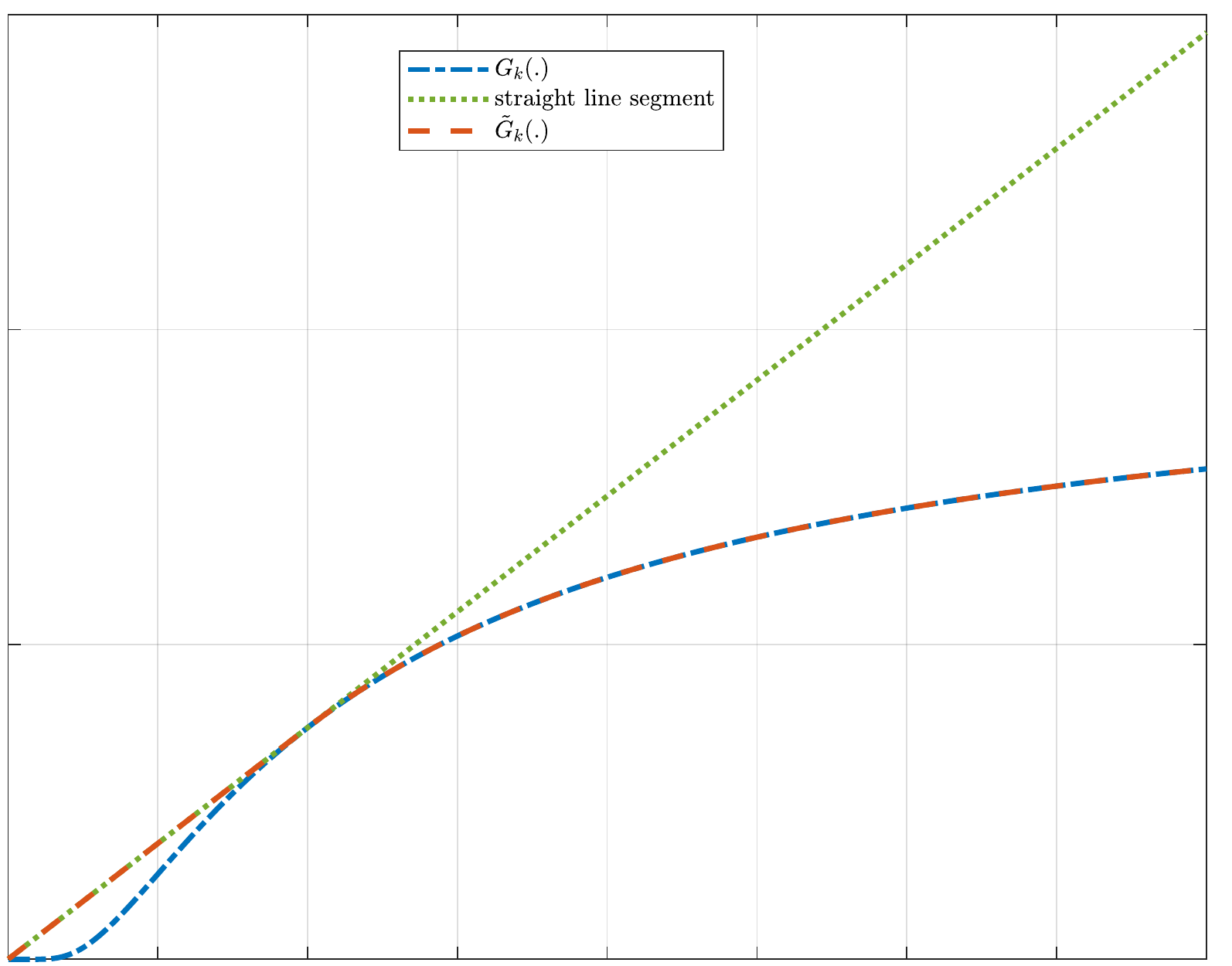}}
 	\caption{Understanding of the convex approximation.}
 	\label{fig0}
 	%\vspace{-1em}
 \end{figure}
 For the sake of understanding, we illustrates an example of the adopted convex approximation in Fig. \ref{fig0}. Then we can solve the following \emph{convex} problem as an approximate of Problem 7:
 \begin{prob}
 	\label{p8}
 \begin{subequations}
 	\begin{align}
 %	\text{\bf{Problem 7:}}
 	\quad  & \max_{\mathbf{\hat p}} ~\sum_{k=1}^K \tilde  G_k(\hat p_{m(k)}) \label{eq:objective-2}
 	\\ 
 	&\text{s.t.}  \quad
 	\sum_{k=1}^{K} (r+1)^{k-1} \hat p_{m(k)} \leq \overline p, \label{eq:constraint:power-2}
 	\\
 	& \quad\quad \hat p_{m(1)} \geq \hat p_{m(2)} \geq \dots \geq \hat p_{m(K)} ,  \label{eq:constraint:order-2}
 	\\
 	&\quad \quad 
 	\hat p_{m(k)} \geq 0, \quad k =1,\dots,K. \label{eq:constraint:positive-2}
 	\end{align}
 \end{subequations}
 \end{prob}

 Let $\mathbf{\hat p}^o = (\hat p_{m(1)}^o,\dots,\hat p_{m(K)}^o)$ be an optimal solution that we obtain by solving Problem \ref{p8}, and denote the optimal objective value of Problem 8 as $\tilde U^o =\tilde U(\mathbf{\hat p}^o) :=\sum_{k=1}^K \tilde G_k(\hat p_{m(k)}^o)$. Note that $\mathbf{\hat p}^o$ is also a \emph{feasible} solution to Problem \ref{p7}. Moreover, denote the objective value of Problem \ref{p7} at $\mathbf{\hat p}^o$ as $U^o = U(\mathbf{\hat p}^o) =\sum_{k=1}^K G_k(\hat p_{m(k)}^o)$. Then the optimal objective value of Problem \ref{p7}, denoted by $U^*$, is bounded as $U^o \leq U^* \leq \tilde U^o$.  
 The following Corollary provides an upper bound of the suboptimality gap $U^* - U^o$ for Problem \ref{p7}. 
 \begin{corollary}\label{c1}
 The gap between the optimal objective value of Problem \ref{p7} and that of Problem \ref{p8} is bounded by $e^{-2} \sum_{k=1}^K w_{m(k)}\Delta_{m(k)}$. Mathematically, $U^*-U^o\leq e^{-2} \sum_{k=1}^K w_{m(k)}\Delta_{m(k)}$.
 \end{corollary}
\begin{proof}
See the Appendix \ref{pc1}.
\end{proof}
 %The following analysis provides an upper bound of $\tilde U^o - U^o$, which is also an upper bound of the suboptimality gap $U^* - U^o$ for Problem \ref{p7}.

We realize that it could be difficult to derive the closed-form solution to both Problem \ref{p7} and Problem \ref{p8}. However, compared to Problem 7, Problem \ref{p8} can be solved efficiently via any convex optimization solver. Besides, Corollary \ref{c1} offers the upper bound of the suboptimality gap between Problem 8 and Problem \ref{p7}.

Moreover, for a fixed total number $N$ of clients and a fixed number $K$ of clients to be served, there are in total $C_N^K$ possible sequences ${\{m(k)\}}_{k=1,2,...K}$. By traversing all these combinations, we can find the optimal solution to Problem \ref{p8} with the optimal set of $K$ clients to be served $\mathcal{K}={\{m(k)\}}_{k=1,2,...K}$. It is worth emphasizing that we traverse all these combinations by substituting them to \eqref{eq:objective-1} rather than \eqref{eq:objective-2}, and then select the one with the maximum objective value.

%In addition, by analyzing $G_k(\hat p_{m(k)})$ and $\tilde  G_k(\hat p_{m(k)})$, we can see that when $\hat p_{i}=\hat p_{j}$, $\forall i<j$, $i,j\in\{1,2,...,N\}$, if $w_i\Delta_{i}\leq w_j\Delta_{j}$, then $G_k(\hat p_{i})\leq G_k(\hat p_{j})$ and $\tilde  G_k(\hat p_{i})\leq \tilde  G_k(\hat p_{j})$ as $d_i^\tau>d_j^\tau$. Thus, when $w_i\Delta_{i}\leq w_j\Delta_{j}$, $\forall i<j$, $i,j\in\{1,2,...,N\}$, if client $i$ is in the set $\mathcal{K}$, then client $j$ will also be included in $\mathcal{K}$. 
\subsubsection{\textbf{Step 2: Optimal number of clients to be served}} By comparing the optimal performance for every $K\in \{1,2,...,N\}$, we can find the optimal value $K^*$, and its corresponding clients to be served $\mathcal{K}^*={\{m(k)\}}_{k=1,2,...K^*}$ and $\left(\hat p_{m(1)}, \hat p_{m(1)},  \hat p_{m(3)}, \dots, \hat p_{m(K^*)}\right)$.  It is worth emphasizing that we traverse all $K\in \{1,2,...,N\}$ by substituting them to the object in Problem \ref{p1*} to find $\mathcal{K}^*={\{m(k)\}}_{k=1,2,...K^*}$ and the corresponding value $\left(\hat p_{m(1)}, \hat p_{m(1)},  \hat p_{m(3)}, \dots, \hat p_{m(K^*)}\right)$. Then, according to the relationship between $\{p_{m(k)}\}$ and $\{\hat p_{m(k)}\}$, we can transfer $\{\hat p_{m(k)}\}$ to the power allocated to each client, and obtain $\{p_{m(k)}\}$ and $p_i=0$, if $i\notin \mathcal{K}^*$.

To summarize our method, the pseudocode of the overall algorithm for resolving Problem 3 is described in Algorithm \ref{alg1}.
%The pseudocode of the algorithm that calculates power allocated to each client is given in Algorithm \ref{alg1}.
\begin{small}
\begin{algorithm}[!htbp]
	\caption{Calculate power allocated to each client}
	\begin{algorithmic}[1]
		\Require ~~ \\
		\textbf{Input}:$\mathbf{\Delta}(t)=\{\Delta_1(t),\Delta_2(t),..., \Delta_N(t)\}$, $(d_1,d_2,...,d_N)$, $r$, $\tau$ and $\sigma^2$.
		\For{$K=1$ to $N$}
		\State $\eta_K=0$;
		\For{$j=1$ to $C_N^K$}
		\State ${\{m(k)\}}_{k=1,2,\cdots,K}={\{m_j(k)\}}_{k=1,2,\cdots,K}$;\Comment{The subset of $\mathcal{N}$ with $K$ clients}
		\State $\{\hat{p}_{m(k)}\}$:= solution to Problem \ref{p8};\Comment{Solve Problem \ref{p8} by convex optimization tool.}
		\If{$\eta_K<\sum_{k=1}^{K}G_k(\hat{p}_{m(k)})$}
		\State $\eta_K=\sum_{k=1}^{K}G_k(\hat{p}_{m(k)})$;
		\State ${\{m_K^*(k)\}}_{k=1,2,\cdots,K}={\{m(k)\}}_{k=1,2,\cdots,K}$;
		\State $\mathbf{\hat p}_K={\{\hat{p}_{m(k)}\}}_{k=1,2,\cdots,K}$;
		\EndIf
		\EndFor
		\EndFor
		\State $K^*=\arg \underset{K=1,2,...,N}{\max} \eta_K$;
		\State $\mathcal{K}^*={\{m_{K^*}^*(k)\}}_{k=1,2,\cdots,K^*}$;\Comment{The set of served clients}
		\State convert $\mathbf{\hat p}_{K^*}$ to $\mathbf{ p}_{K^*}$ using \eqref{phat}; \Comment{Power allocated to clients in $\mathcal{K}^*$.}
	\end{algorithmic}
	\label{alg1}
\end{algorithm}
\end{small}
%\begin{figure}[!t]
%	\centerline{\includegraphics[width=0.25\textwidth]{policy2}}
%	\caption{policy}
%	\label{fig1}
%\end{figure}
\section{Numerical Results and Discussions}
In this section, simulation results are provided to evaluate the effectiveness of the proposed adaptive NOMA/OMA scheme for both two-client and multi-client scenarios. 
\subsection{Two-client scenario}
This subsection provides numerical results to verify the analytical results for the two-client scenario presented in Section III. We set path loss exponent $\tau=2$ and the target data rate $R=1$ in all simulations. The SNR in this subsection refers to the transmission SNR $\rho$.

We follow \cite{sennott2009stochastic} and apply Relative Value Iteration (RVI) method on truncated finite states ($\Delta_i \leq 100$, $\forall i$) to approximate the countable infinite state space. The optimal policy and suboptimal policy is illustrated in Fig.\ref{fig1}, where SNR$=18$dB, the normalized distances for two clients are $d_1=2$ and $d_2=4$, and the weighted parameters for two clients $w_1=w_2=0.5$. We can observe the switching structure of the optimal policy which verifies Theorem \ref{T2}. Besides, we can find that the proposed suboptimal policy is similar to the optimal policy.
\begin{figure}[!tbp]
	\centering
	\begin{subfigure}[b]{0.4\textwidth}
		\centering
		\includegraphics[width=\textwidth]{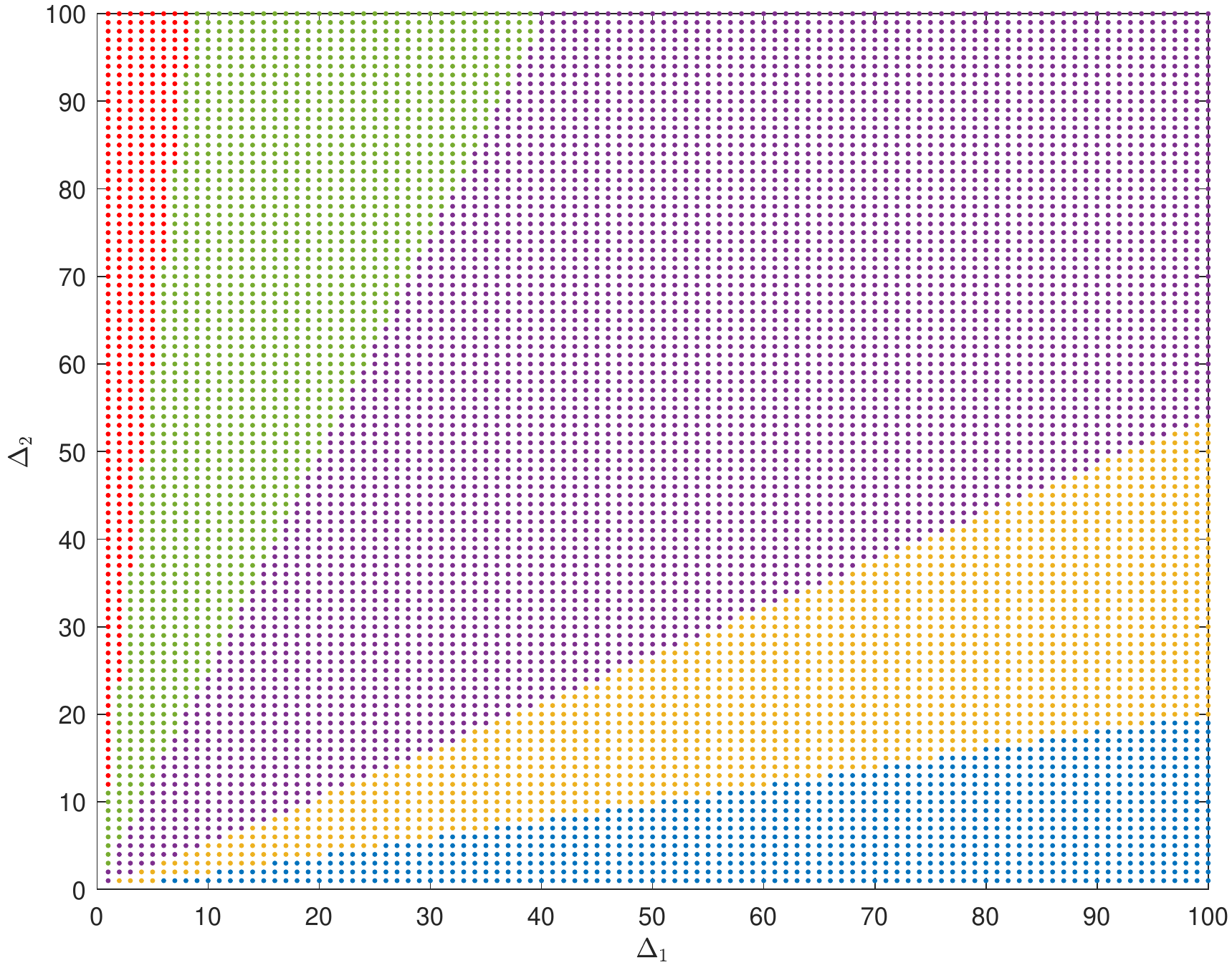}
		\caption{MDP optimal policy}
		\label{fig10}
	\end{subfigure}
	\hfill
	\begin{subfigure}[b]{0.4\textwidth}
		\centering
		\includegraphics[width=\textwidth]{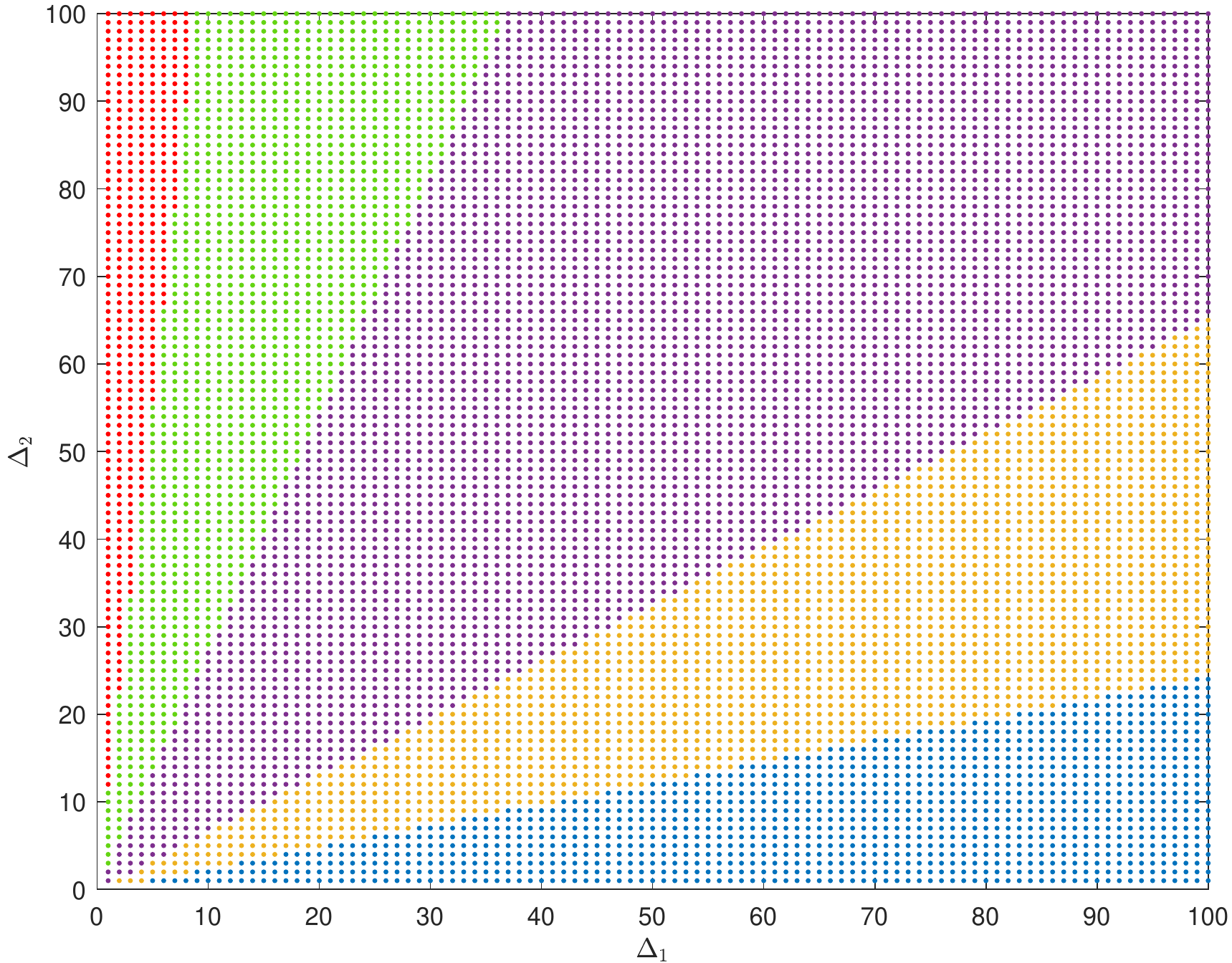}
		\caption{Suboptimal policy}
		\label{fig11}
	\end{subfigure}
	\caption{Age-optimal policy and suboptimal policy. Each point represents a state $s=(\Delta_1,\Delta_2)$. The colored area indicates action for each state, i.e., $a=0$ for states in the blue area; $a=7$ for states in the orange area; $a=8$ for states in the purple area; $a=9$ for states in the green area and $a=10$ for states in the red area, where $L=10$ and $\mathcal{A}=\{0,6,7,8,9,10\}$.}
	\label{fig1}
\end{figure}
\begin{figure}[!t]
	%\centerline{\includegraphics[width=0.4\textwidth]{newfinalresult2}}
		\centerline{\includegraphics[width=0.43\textwidth]{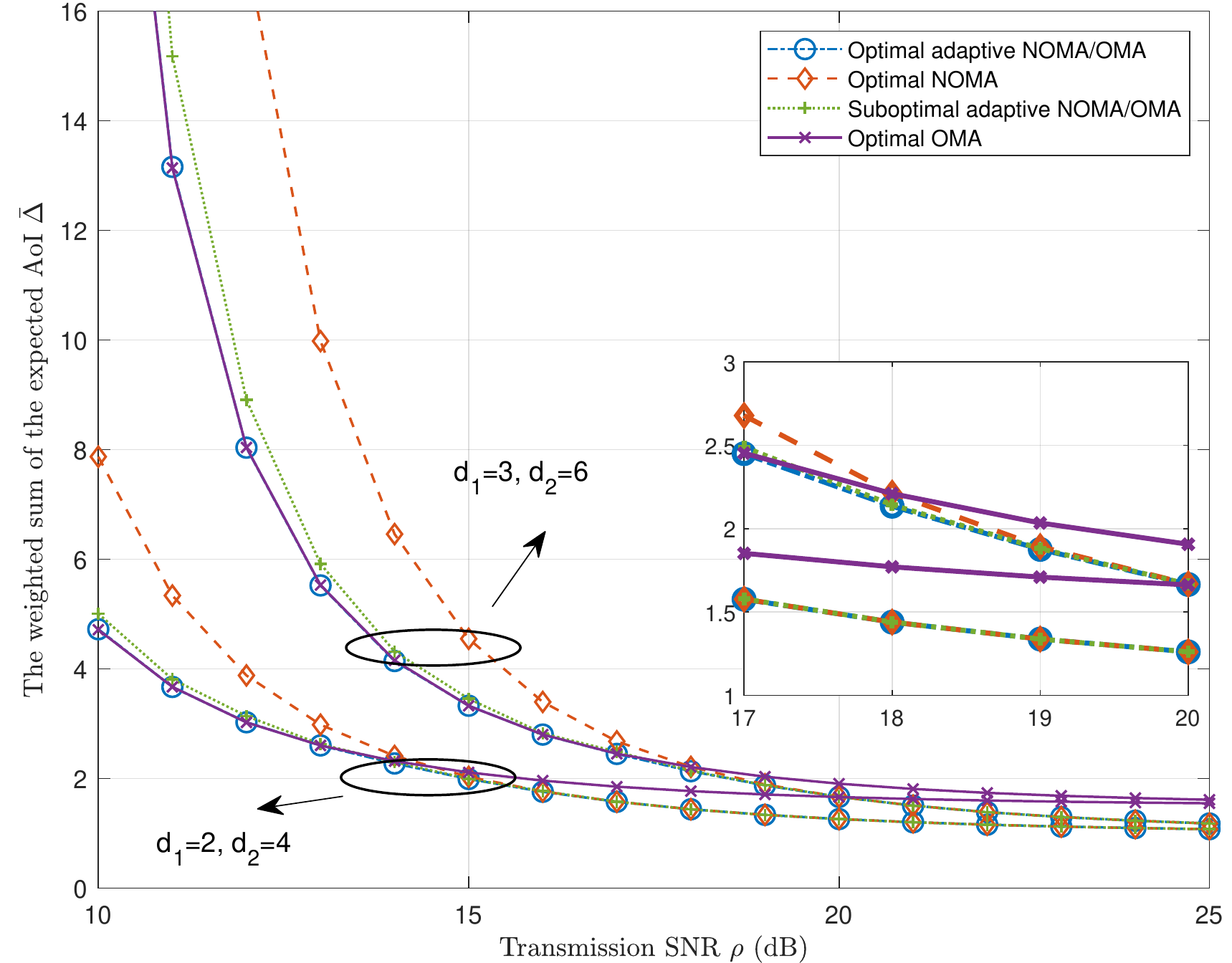}}
	\caption{The performance comparison of different policies versus SNR for the two-client scenario with $w_1=w_2=0.5$.}
	\label{fig2}
	%\vspace{-1em}
\end{figure}

Fig. \ref{fig2} compares the weighted sum of the expected AoI of the two clients under optimal policy using adaptive NOMA/OMA scheme (optimal adaptive NOMA/OMA scheme), the policy that always using NOMA for transmission (optimal NOMA policy with $a\in\{\max\{\lceil{\frac{L}{2}\rceil}+1,\lceil\frac{(2^R-1)L}{2^R}\rceil\},...,L-1\}$), the proposed suboptimal policy and the optimal OMA policy that the BS adaptively selects one client to conduct transmission (optimal OMA scheme with $a\in \{0,L\}$) in two cases: 1) $d_1=2$ and $d_2=4$; 2) $d_1=3$ and $d_2=6$. The setting of the rest system parameters is the same as that in Fig \ref{fig1}. We conduct the simulations by generating ${10}^6$ time slots for different transmission SNRs. We can see from Fig.\ref{fig2} that the proposed suboptimal policy achieves near-optimal performance: its weighted sum of the expected AoI almost coincides with that of the optimal adaptive NOMA/OMA policy especially when the outage probability of two clients are small as shown in Fig. \ref{fig2}. Specifically, the performance of suboptimal policy is closer to that of the optimal adaptive NOMA/OMA policy when $d_1=2$ and $d_2=4$, comparing to the case when $d_1=3$ and $d_2=6$; the gap between the AoI performance of the suboptimal policy and that of the optimal adaptive NOMA/OMA policy narrows as the SNR increases.% weighted sum of the expected AoI

Moreover, we can see that when SNR is small, e.g., SNR$< 15$dB, the performance of optimal adaptive NOMA/OMA scheme and that of the optimal OMA scheme are almost the same in Fig. \ref{fig2}. This is due to the low SNR, which leads to a higher outage probability for both OMA and NOMA. The situation for NOMA is even worse. As such, both optimal adaptive NOMA/OMA policy and the suboptimal policy will prefer not to choose NOMA scheme but use OMA scheme. Thus, these two policies have similar performance. As SNR increases, the weighted sum of the expected AoI of optimal OMA policy will approach $1.5$, when $w_1=w_2=0.5$. This is the optimal performance under the OMA scheme. As the outage probability of each client is approaching 0, the instantaneous age of each client will equal to 1 and 2 iteratively.% {\color{red}age $1$ and $2$ take turns to form the age evolution of each client. }

Furthermore, we can see from Fig. \ref{fig2} that the performance of optimal adaptive NOMA/OMA policy and that of suboptimal policy and NOMA policy are relatively close when SNR is large, e.g., SNR$\geq20$dB. This is because both optimal adaptive NOMA/OMA policy and suboptimal policy are more likely to choose NOMA for transmission to both clients at the same time. When SNR is large enough, the optimal performance of both the optimal adaptive NOMA/OMA policy and the suboptimal policy approaches $1$ as the instantaneous AoI of each client will be always $1$, thanks to almost no outage for both clients in NOMA at high SNR. The BS thus always chooses NOMA scheme to conduct transmissions to both clients. In addition, NOMA is better than optimal OMA when SNR$>16$dB for $d_1=2$ and $d_2=4$ and SNR$>19$dB for  $d_1=3$ and $d_2=6$. This shows the benefits of NOMA in timely status update when SNR is large.

\subsection{Multi-client scenario}
In this subsection, we evaluate the effectiveness of approximation of the max-weight policy in multi-client scenario. We conduct all simulations by generating ${10}^5$ time slots for different transmission SNR $\rho=\bar{p}/{\sigma^2}$. We consider the scenario with a BS conducting timely status update to $5$ clients with normalized distance $d_i=6-i$, $i\in\{1,2,..,5\}$. We set path loss exponent $\tau=2$ and the target data rate $R=1$. Fig. \ref{fig3*} illustrates the performance of different policies under different transmission SNR, including: 1) max-weight policy under adaptive NOMA/OMA solved by exhaustive search in each time slot (MW-N/OMA), 2) approximated convex optimization policy (termed AP-N/OMA), 3) approximated convex optimization policy under NOMA with fixed client number $K$ (termed AP-NOMA-F-$K$) and 4)  OMA scheme that selects the client corresponding to achieve maximum expected age drop to serve as in \cite{kadota2019minimizing} (termed MW-OMA). 

We can see that similar to the results of the two-client scenario, when the SNR is low, the AoI performance under different NOMA schemes (i.e., AP-NOMA-F-$K$) is poor, due the relatively large outage probability of NOMA scheme in low SNR scenario, comparing with MW-OMA scheme. Specifically, when SNR $\rho\leq 13$dB, the performance of AP-NOMA-F-$K_1$ is worse than that of AP-NOMA-F-$K_2$, if $K_1>K_2$. As the transmission SNR increases, the performance of AP-NOMA-F-$K$ becomes better. The rationale is that when the transmission SNR is sufficiently large, the NOMA scheme that allows to serve more clients achieves reduced age performance. When SNR $\rho\geq 29$dB, the performance of AP-NOMA-F-$K_1$ is better than that of AP-NOMA-F-$K_2$, if $K_1>K_2$. Comparing to the AP-NOMA-F-$K$ and MW-OMA, the proposed AP-N/OMA scheme that adaptively switches between NOMA and OMA achieves overall better AoI performance as it allocates power to each client in a more flexible way. In addition, the small gap between MW-N/OMA policy and AP-N/OMA shows the effectiveness of our proposed approximation method which reduces the computation complexity but achieves near-optimal performance. 
  
%\begin{figure}[!htbp]
%	\centerline{\includegraphics[width=0.5\textwidth]{resultmd2}}
%	\caption{The performance comparison of different policies versus SNR for multi-client scenario, $N=5$ with $w_i=1/N$, $\forall i\in\{1,2,..,5\}$.}
%	\label{fig3}
%	%\vspace{-1em}
%\end{figure}
\begin{figure}[!tbp]
	\centerline{\includegraphics[width=0.43\textwidth]{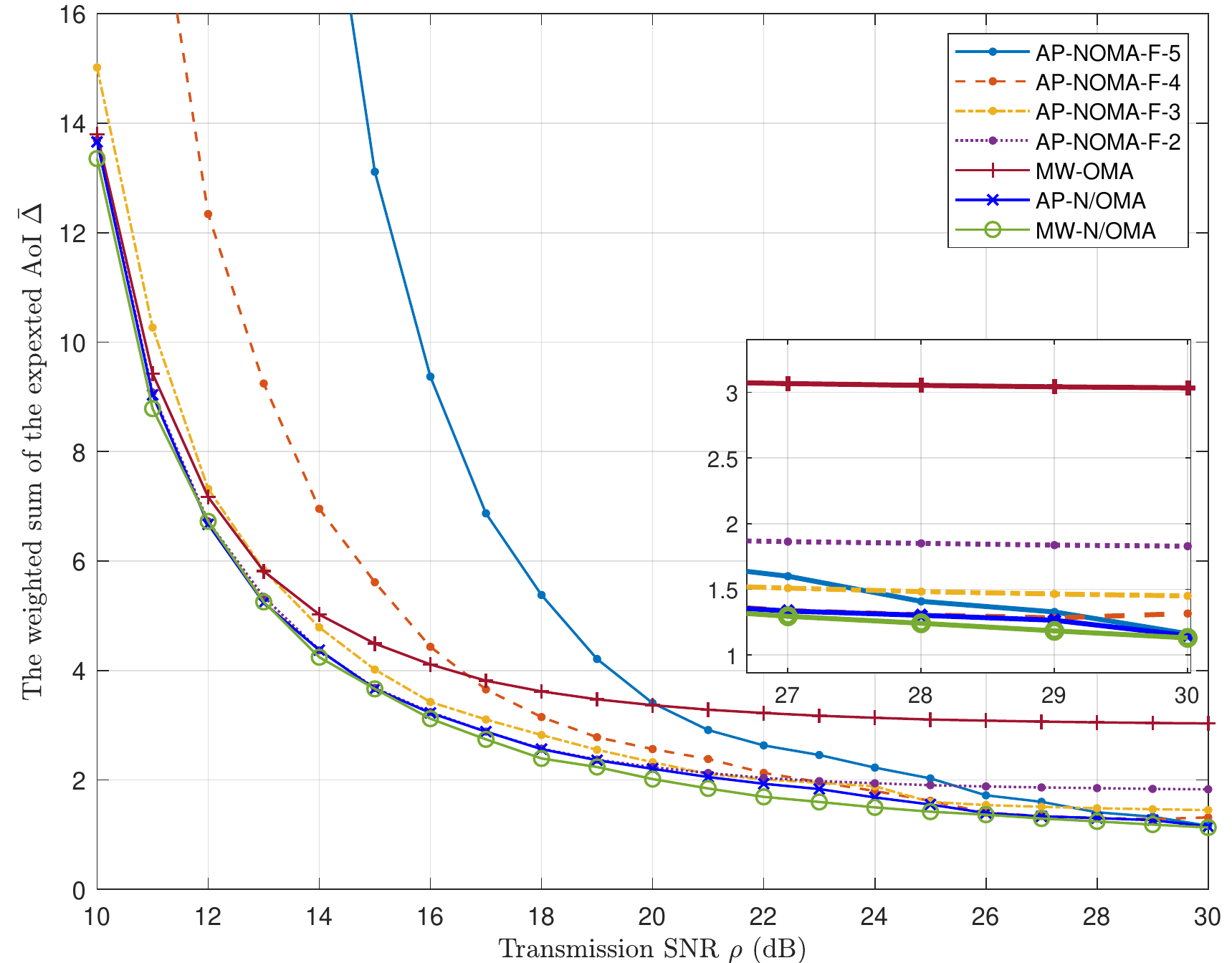}}
	\caption{The performance comparison of different policies versus SNR for multi-client scenario, $N=5$ with $w_i=1/N$, $\forall i\in\{1,2,..,5\}$.}
	\label{fig3*}
	%\vspace{-1em}
\end{figure}

Fig. \ref{fig4} plots the curves of the weighted sum of expected AoI performance for MW-OMA policy, AP-N/OMA policy and MW-N/OMA policy versus the number of clients in the network. The network with increasing number of clients is considered with $N\in\{2,3,4,5,6\}$, the normalized distance of $i$th client in the system with $N$ clients is $d_i=N+1-i$ and weighted parameter $w_i=1/N$. As shown in Fig. \ref{fig4}, the performance of AP-N/OMA scheme is close to that of MW-N/OMA. Moreover, comparing to the MW-OMA scheme, it achieves significant performance improvement. Besides, the AP-N/OMA scheme has a slow speed of AoI increase due to the increasing number of clients in the network, comparing with MW-OMA scheme. The performance gap between MW-OMA and AP-N/OMA and that between MW-OMA and MW-N/OMA, both increase as the number of clients in the network increases. This shows the potential of adaptive NOMA/OMA scheme in achieving reduced AoI performance for multi-client network. The rationale behind is that in MW-OMA scheme, only one client can be served to have potential AoI drop while other clients' AoI will certainly increase. The increasing number of clients in the network makes more clients have AoI increase. Thus, the age of network will increase. While for adaptive NOMA/OMA, as more than one client can be served at each time slot, the speed of AoI increase due to the increasing number of clients in the network will slow down.

\begin{figure}[!tbp]
	\centerline{\includegraphics[width=0.4\textwidth]{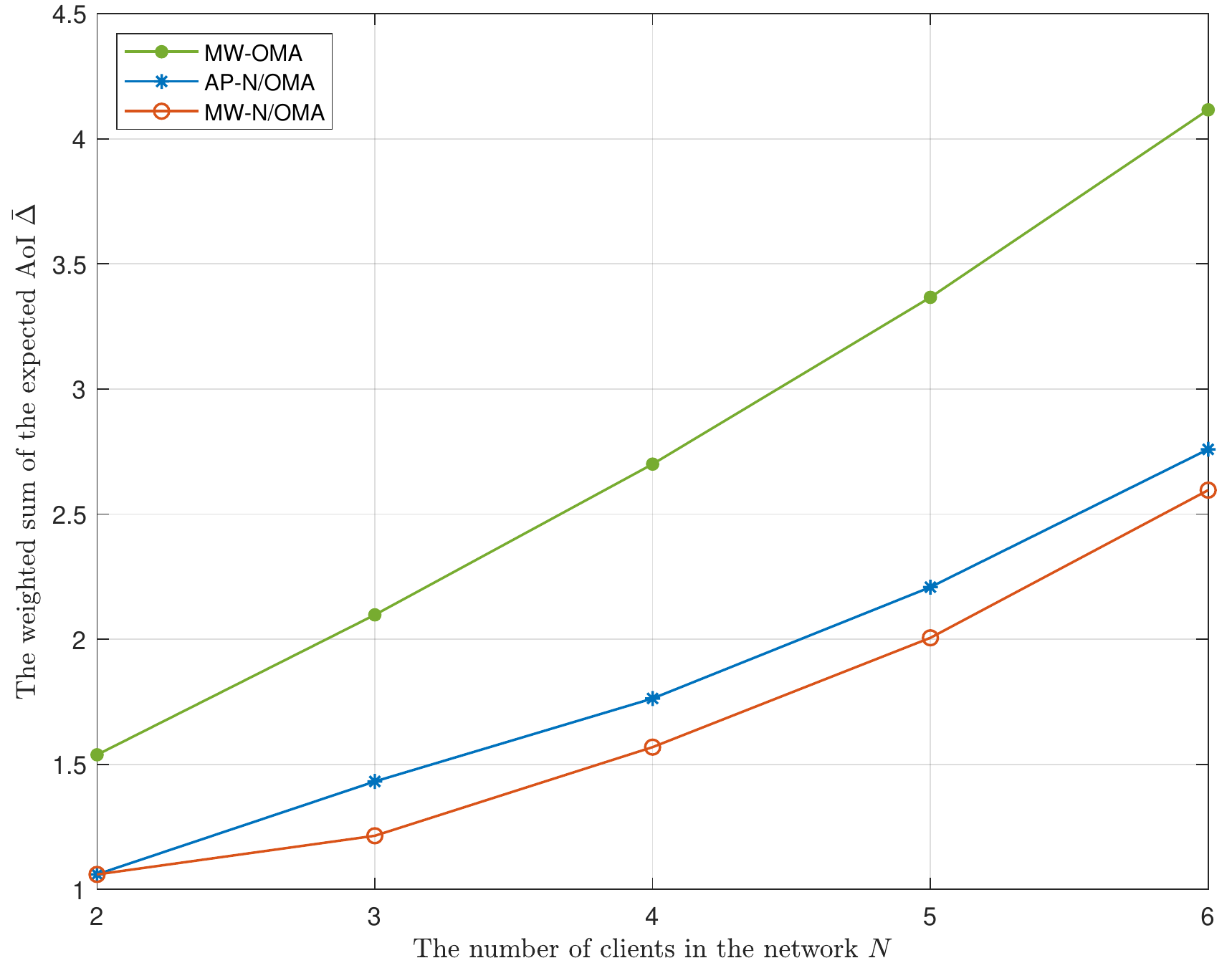}}
	\caption{Simulation of network with different number of clients $N$ with $w_i=1/N$, $\forall i$ and transmission SNR $\rho=20$ dB.}
	\label{fig4}
	%\vspace{-1em}
\end{figure}
\section{Conclusions}
In this paper, we considered a wireless network with a base station (BS) conducting timely transmission to multiple clients in a time-slotted manner. The BS can adaptively switch between NOMA and OMA for the downlink transmission to minimize the AoI of the network. We studied both two-client scenario and multi-client scenario. For the two-client scenario, we developed an optimal policy for the BS to decide whether to use NOMA or OMA for downlink transmission based on the instantaneous AoI of both clients in order to minimize the weighted sum of the expected AoI of the network. This was achieved by formulating and resolving a Markov Decision Process (MDP) problem. We proved the existence of an optimal stationary and deterministic policy. Action elimination was conducted to reduce the computation complexity. The optimal policy is shown to have a switching-type property with obvious decision boundaries. A suboptimal policy with lower computation complexity was also proposed, which is shown to achieve near-optimal performance according to simulation results. 

For the multi-client scenario, inspired by the proposed near-optimal policy, we formulated a nonlinear optimization problem to determine the optimal power allocated to each client by maximizing the expected AoI drop of the network in each time slot. We managed to resolve the formulated problem by approximating it as a convex optimization problem. The upper bound of the gap between the approximate convex problem and the original nonlinear, nonconvex problem was derived. Simulation results validated the effectiveness of the approximation. The performance adaptive NOMA/OMA scheme by solving the convex optimization was shown to be close to that of max-weight policy solved by exhaustive search. Besides, the adaptive NOMA/OMA scheme has achieved significantly reduced AoI comparing to OMA scheme, especially when the number of clients in the network is large and the transmission SNR is high.
%In this paper, we considered a wireless network with a base station (BS) conducting timely transmission to two clients in a slotted manner via hybrid non-orthogonal multiple access (NOMA)/orthogonal multiple access (OMA). The BS can adaptively switch between NOMA and OMA for the downlink transmission to minimize the AoI of the network. We develop an optimal policy for the BS to decide whether to use NOMA or OMA for downlink transmission based on the instantaneous AoI of both clients in order to minimize the weighted sum of the expected AoI of the network. This is achieved by formulating a Markov Decision Process (MDP) problem. We proved the existence of optimal stationary and deterministic policy. Action elimination was conducted to reduce the computation complexity. The optimal policy is shown to have a switching-type property with obvious decision boundaries. A suboptimal policy with lower computation complexity was also proposed, which is shown to achieve near-optimal performance according to simulation results. The approximate average AoI performance of the suboptimal policy was also derived. The performance of different policies under different system settings was compared and analyzed in numerical results to provide useful insights for practical system designs. 
%When choosing OMA, the BS decides which client to serve. When using NOMA, the BS determines the amount of power allocated to each client.
\begin{appendices}
	\section{Proof of Theorem \ref{TE}}%[Proof of Theorem \ref{TE}]
	\label{A1}
	We prove this theorem by verifying Assumptions 3.1, 3.2 and 3.3 in \cite{guo2006average} hold. As the action space for each state is finite, Assumption 3.2 holds, and we only need to verify the following two conditions. 
	\begin{itemize}
		\item[\text{1)}] There exist positive constants $\beta < 1$, $M$ and $m$, and a measurable function $\omega(s) \geq1$, $s=(\Delta_1,\Delta_2)\in \mathcal{S}$ such that the reward function of MDP problem $r(s,a)=w_1\Delta_1+w_2\Delta_2$, $|r(s,a)|\leq M\omega(s)$ for all state-action pairs $(s,a)$ and
	%	\begin{small}
			\begin{equation}
			\sum_{\hat{s}\in S} \omega(\hat{s})P(\hat{s}|s,a)\leq\beta\omega(s)+m,\ {\rm \ for \ all}\ (s,a).
		%	\vspace{-.5em}
			\end{equation}
	%	\end{small} 
		\item [\text{2)}] There exist two value functions $v_1,v_2 \in B_{\omega}(\mathcal{S})$, and some state $s_0\in \mathcal{S}$, such that 
		%\vspace{-.5em}
			\begin{equation}
			v_1(s)\leq h_{\alpha}(s)\leq v_2(s), \ {\rm for \ all} \ s\in \mathcal{S},\ {\rm and} \ \alpha \in(0,1),
		%	\vspace{-.5em}
			\end{equation} where $h_{\alpha}(s)=V_{\alpha}(s)-V_{\alpha}(s_0)$ and $B_{\omega}(\mathcal{S}):=\{u:\Vert u\Vert_{\omega} <\infty \}$ denotes Banach space, $\Vert u\Vert_{\omega}:=\sup_{s\in S}\omega(s)^{-1}|u(s)|$ denotes the weighted supremum norm.
	\end{itemize}

For condition 1, we show that when $\omega(s)=w_1\Delta_1+w_2\Delta_2$ and $m >1 $, there exists a $\beta$ that  $\underset{a}{\max}\{\frac{w_1\Delta_1 P_1^O+w_2\Delta_2+1-m}{w_1\Delta_1+w_2\Delta_2},\frac{w_1\Delta_1+w_2P_2^O\Delta_2+1-m}{w_1\Delta_1+w_2\Delta_2},$ $\;\frac{w_1P_1^N(a)\Delta_1+w_2P_2^N(a)\Delta_2+1-m}{w_1\Delta_1+w_2\Delta_2}\}\leq \beta<1$ to meet condition 1. To prove condition 2 in our problem, we show that when $\omega(s)=w_1\Delta_1+w_2\Delta_2$, there exists $\frac{w_1\Delta_1+w_2\Delta_2+1}{w_1\Delta_1+w_2\Delta_2}\leq \kappa<\infty$ that $\sum_{\hat{s}\in S} \omega(\hat{s})P(\hat{s}|s,a)\leq\kappa\omega(s) $ for all $(s,a)$, and for $d\in D^{MD}$, where $D^{MD}$ denotes the set of Markovian
and deterministic decision rule, $ \sum_{\hat{s}\in S} \omega(\hat{s})P_d(\hat{s}|s,a)\leq \omega(s)+1\leq (1+1)\omega(s)$, so that  $\alpha^T\sum_{\hat{s}\in S} \omega(\hat{s})P_d^T(\hat{s}|s,a)\leq \alpha^T (\omega(s)+T)<\alpha^T(1+T)\omega(s)$. Hence, for each $\alpha$, $0\leq \alpha <1$, there exists a $\eta$, $0\leq \eta<1$ and an integer $T$ such that 
%	\vspace{-.5em}
	\begin{equation}
	\alpha^T\sum_{\hat{s}\in S} \omega(\hat{s})P_{\pi}^T(\hat{s}|s,a)\leq \eta \omega(s)
%	\vspace{-.5em}
	\end{equation} for $\pi=(d_1,...,d_T)$, where $d_k\in D^{MD}$, $1\leq k\leq T$. Then, according to Proposition 6.10.1 \cite{puterman2014markov}, for each $\pi\in \Pi^{MD}$, where $\Pi^{MD}$ denotes the set of Markovian deterministic policies, and $s\in S$
%	\vspace{-.5em}
		\begin{equation}
		|V_{\alpha}(s)|\leq \frac{1}{1-\eta}[1+\alpha\kappa+...+(\alpha\kappa)^{(T-1)}]w(s).
	%	\vspace{-.5em}
		\end{equation} We thus can further prove condition 2. This completes the proof.
\section{Proof of Theorem \ref{T2}}%[Proof of Theorem \ref{TE}]
\label{A2}
The switching-type policy is actually the same as the monotonically nondecreasing policy in $\Delta_2$ when $\Delta_1$ is fixed, and the monotonically nonincreasing policy in $\Delta_1$ when $\Delta_2$ is fixed. To prove the monotonicity of the optimal policy of the MDP problem in $\Delta_2$, we verify that the following four conditions given in \cite[Theorem~8.11.3]{puterman2014markov} hold.
\begin{itemize}
	\item [a)] The reward function $r(s,a)$ is nondecreasing in $s$ for all $a\in \mathcal{A}$;
	\item [b)] $q(k|s,a)=\sum_{j=k}^{\infty}p(j|s,a)$ is nondecreasing in $s$ for all $k\in \mathcal{S}$ and $a\in \mathcal{A}$, where $p(j|s,a)$ is the state transition probability $P(s_{t+1}=j|s_t=s,a_t=a)$, given in \eqref{e3} and \eqref{e3*};
	\item [c)] $r(s,a)$ is a subadditive function on $\mathcal{S}\times \mathcal{A}$ and
	\item [d)] $q(k|s,a)$ is a subadditive function on $\mathcal{S}\times \mathcal{A}$ for all $k\in \mathcal{S}$.
\end{itemize}
To verify these conditions, we first order the state by $\Delta_2$, i.e., $s^+\geq s^-$ if $\Delta_2^+\geq \Delta_2^-$ where $s^+=(\cdot,\Delta_2^+)$ and $s^-=(\cdot,\Delta_2^-)$. The one-step reward function of the MDP is 
\begin{equation}
\label{req1}
r(s,a)=w_1\Delta_1+w_2\Delta_2.
\end{equation} It is obvious that the condition a) is satisfied. According to the transition probabilities in \eqref{e3} and \eqref{e3*}, if the current state $s=(\Delta_1,\Delta_2)$, the next possible states are $s_1=(\cdot,\Delta_2+1)$ (including $(1,\Delta_2+1)$ and $(\Delta_1+1,\Delta_2+1)$) and $s_2=(\cdot,1)$ (including $(1,1)$ and $(\Delta_1+1,1)$).  Based on \eqref{e3} and \eqref{e3*}, we have
\begin{equation}
\label{req2}
q(k|s,a=0)=\left\{
\begin{array}{rcl}
0,& &\text{if }  k > s_1 \\
1,& &\text{otherwise.}  \\
\end{array}
\right.
\end{equation}
\begin{equation}
\label{req4}
q(k|s,a=i, 0<i<L)=\left\{
\begin{array}{rcl}
0,& &\text{if }  k > s_1\\
P_2^N(i),& &\text{if }  s_1\geq k>s_2 \\
1,& & \text{if } k\leq s_2
\end{array}
\right.
\end{equation}
\begin{equation}
\label{req3}
q(k|s,a=L)=\left\{
\begin{array}{rcl}
0,& &\text{if }  k > s_1\\
P_2^O,& &\text{if }  s_1\geq k>s_2 \\
1,& & \text{if } k\leq s_2
\end{array}
\right.
\end{equation} Thus, condition b) is immediate. 

To verify the remaining two conditions, we give the definition of subadditivity in the following
\begin{definition}
	\label{d0}
	(Subadditivity\cite{puterman2014markov}) A multivariable function $Q(\delta,a): \mathcal{S} \times \mathcal{A} \rightarrow \mathbb{R}$ is subadditive in $(\delta,a)$ , if for all $\delta^{+}\geq\delta^{-}$ and $a^{+}\geq a^{-}$,
	\begin{equation}
	\label{e8}
	\begin{aligned}
	Q(\delta^{+},a^{+})+ Q(\delta^{-},a^{-}) \leq Q(\delta^{+},a^{-})+ Q(\delta^{-},a^{+})
	\end{aligned}
	\end{equation}holds.
\end{definition} 
According to \eqref{req1}, condition c) follows. For the last condition, we verify whether
\begin{equation}
q(k|s^+,a^+)+q(k|s^-,a^-)\leq q(k|s^+,a^-)+q(k|s^-,a^+),
\end{equation} with $s^+=(\Delta_1,\Delta_2^+)$ and $s^-=(\Delta_1,\Delta_2^-)$ where $\Delta_2^+\geq \Delta_2^-$ and $a^+\geq a^-$. As there are three actions, we consider three cases: (1) $a^+=i$, $a^-=0$, (2) $a^+=N$, $a^-=0$ and (3) $a^+=N$, $a^-=i$ and (4) $a^+=i$, $a^-=j$ for  $ 0\leq j\leq i \leq L$, $\forall i,j$.
According to \eqref{req2}-\eqref{req3}, we can verify that condition d) holds. As all these four conditions hold, the optimal policy is monotonically nondecreasing in $\Delta_2$, when $\Delta_1$ is fixed. The proof of monotonicity of the optimal policy of the MDP problem in $\Delta_1$ is similar, thus omitted for brevity. This completes the proof.
\section{Proof of Lemma \ref{lema1}}
\label{A3}
 Consider any feasible point $(\hat p_{m(1)},\dots,\hat p_{m(K)})$ of Problem 6. Suppose $\hat p_{m(1)} < \hat p_{m(2)}$.
By decreasing $\hat p_{m(2)}$ to the same value as $\hat p_{m(1)}$, we construct another point:
\begin{eqnarray}
\hat p' = (\hat p_{m(1)}' , \hat p_{m(2)}', \hat p_{m(3)}' \dots,\hat p_{m(K)}' ) 
= (\hat p_{m(1)}, \hat p_{m(1)},  \hat p_{m(3)}, \dots, \hat p_{m(K)}) %\nonumber
\end{eqnarray}
which is still feasible in terms of \eqref{eq:constraint:power}--\eqref{eq:constraint:positive}. Moreover, it can be verified that
\begin{eqnarray}
\max\left\{\frac{1}{\hat p_{m(1)}'}, \dots, \frac{1}{\hat p_{m(k)}'} \right\}= \max\left\{\frac{1}{\hat p_{m(1)}}, \dots, \frac{1}{\hat p_{m(k)}}\right\} %\nonumber
\end{eqnarray}
for all $k=1,...,K$. Hence the optimal objective value \eqref{eq:objective} will not change if we add constraint $\hat p_{m(2)} \leq \hat p_{m(1)}$ to Problem 6. The same argument implies that adding $\hat p_{m(k)} \leq \hat p_{m(k-1)}$ for all $k=2,\dots,K$ to Problem 6 will not change its optimal objective value. This completes the proof.
\section{Proof of Corollary \ref{c1}}
\label{pc1}
	 For each $k=1,\dots,K$, there is a unique point $\hat p_{m(k)}' \in (0, \frac{d_{m(k)}^\tau r \sigma^2}{2}]$, such that
\begin{eqnarray}
G'_k(\hat p_{m(k)}') =w_{m(k)}\Delta_{m(k)} \exp\left(- \frac{d_{m(k)}^\tau r \sigma^2}{\hat p_{m(k)}'} \right) \times \frac{d_{m(k)}^\tau r \sigma^2}{(\hat p_{m(k)}')^2 } = \frac{ w_{m(k)}\Delta_{m(k)} e^{-1} }{d_{m(k)}^\tau r \sigma^2}. \label{eq:max_distance_condition}
\end{eqnarray} 
The difference $\tilde G_k(\hat p_{m(k)}) - G_k(\hat p_{m(k)})$ is maximized at $\hat p_{m(k)} = \hat p_{m(k)}'$, which is:
\begin{eqnarray}
&&\tilde G_k(\hat p_{m(k)}') - G_k(\hat p_{m(k)}') = \frac{ w_{m(k)}\Delta_{m(k)} e^{-1} }{d_{m(k)}^\tau r \sigma^2} \hat p_{m(k)}' - w_{m(k)}\Delta_{m(k)} \exp\left(- \frac{d_{m(k)}^\tau r \sigma^2}{\hat p_{m(k)}'} \right)
\nonumber
\\
&& =w_{m(k)} \Delta_{m(k)} \exp\left(- \frac{d_{m(k)}^\tau r \sigma^2}{\hat p_{m(k)}'} \right) \left(\frac{d_{m(k)}^\tau r \sigma^2}{\hat p_{m(k)}'} - 1\right)  =w_{m(k)} \Delta_{m(k)} e^{-\alpha} \left(\alpha - 1\right) \nonumber
\end{eqnarray}
where $\alpha := \frac{d_{m(k)}^\tau r \sigma^2}{\hat p_{m(k)}'} \in [2, +\infty)$, and the second equality utilized \eqref{eq:max_distance_condition}. It is easy to verify that $e^{-\alpha} \left(\alpha - 1\right)$ is monotonically decreasing on $\alpha \in [2, +\infty)$, and therefore its upper bound is attained at $\alpha = 2$, i.e.,
\begin{eqnarray}
\tilde G_k(\hat p_{m(k)}') - G_k(\hat p_{m(k)}') \leq w_{m(k)}\Delta_{m(k)} e^{-2}. \nonumber
\end{eqnarray}
Therefore, we have
\begin{eqnarray}\label{bound}
\tilde U^o - U^o \leq \sum_{k=1}^K \left[\tilde G_k(\hat p_{m(k)}') - G_k(\hat p_{m(k)}')\right] \leq e^{-2} \sum_{k=1}^K w_{m(k)}\Delta_{m(k)}. %\nonumber
\end{eqnarray} The upper bound of $\tilde U^o - U^o$, which is also an upper bound of the gap $U^* - U^o$ for Problem \ref{p7} as $U^o \leq U^* \leq \tilde U^o$.  This completes the proof.
\end{appendices}

\bibliography{ref}
\bibliographystyle{IEEEtran}

\end{document}